\documentclass[11pt,fleqn,a4paper]{article} 

\usepackage{amsmath,amssymb,amsthm,amsfonts} 
\usepackage[mathscr]{eucal}

\usepackage{hyperref}
\hypersetup{colorlinks, linkcolor=blue, citecolor=blue, urlcolor=blue}

\allowdisplaybreaks

\makeatletter
\long\def\@makecaption#1#2{%
	\vskip\abovecaptionskip\footnotesize
	\sbox\@tempboxa{#1. #2}%
	\ifdim \wd\@tempboxa >\hsize
	#1. #2\par
	\else
	\global \@minipagefalse
	\hb@xt@\hsize{\hfil\box\@tempboxa\hfil}%
	\fi
	\vskip\belowcaptionskip}
\makeatother

\newcommand{\p}{\partial}
\newcommand{\sgn}{\mathop{\rm sgn}\nolimits}
\newcommand{\ord}{\mathop{\rm ord}\nolimits}
\newcommand{\const}{{\rm const}}

\newlength{\mylength}
\newcommand{\solution}{\hspace*{-\mylength}\bullet\quad}

\newtheorem{theorem}{Theorem}
\newtheorem{lemma}[theorem]{Lemma}
\newtheorem{corollary}[theorem]{Corollary}
\newtheorem{conjecture}[theorem]{Conjecture}
\newtheorem{proposition}[theorem]{Proposition}
{\theoremstyle{definition}
	\newtheorem{definition}[theorem]{Definition}
	\newtheorem{remark}[theorem]{Remark}
	
}

\newcommand{\todo}[1][\null]{\ensuremath{\clubsuit}}

\newcommand{\noprint}[1]{}
\newcommand{\lsemioplus}{\mathbin{\mbox{$\lefteqn{\hspace{.77ex}\rule{.4pt}{1.2ex}}{\in}$}}}

\begin{document}

\par\noindent {\LARGE\bf
Extended symmetry analysis of (1+2)-dimensional \\
fine Kolmogorov backward equation
\par}

\vspace{5.5mm}\par\noindent{\large
\large Serhii D. Koval$^{\dag\ddag}$ and Roman O. Popovych$^{\S\ddag}$
}

\vspace{5.5mm}\par\noindent{\it\small
$^\dag$Department of Mathematics and Statistics, Memorial University of Newfoundland,\\
$\phantom{^\ddag}$\,St.\ John's (NL) A1C 5S7, Canada
\par}

\vspace{2mm}\par\noindent{\it\small
$^\S$\,Mathematical Institute, Silesian University in Opava, Na Rybn\'\i{}\v{c}ku 1, 746 01 Opava, Czech Republic
\par}
\vspace{2mm}\par\noindent{\it\small	
$^\ddag$\,Institute of Mathematics of NAS of Ukraine, 3 Tereshchenkivska Str., 01024 Kyiv, Ukraine
\par}

\vspace{4mm}\par\noindent
E-mails:
skoval@mun.ca,
rop@imath.kiev.ua

\vspace{8mm}\par\noindent\hspace*{10mm}\parbox{140mm}{\small\looseness=-1
Within the class of (1+2)-dimensional ultraparabolic linear equations,
we distinguish a fine Kolmogorov backward equation with a quadratic diffusivity.
Modulo the point equivalence, it is a unique equation within the class
whose essential Lie invariance algebra is five-dimensional and nonsolvable.
Using the direct method, we compute the point symmetry pseudogroup of this equation and analyze its structure.
In particular, we single out its essential subgroup and classify its discrete elements.
We exhaustively classify all subalgebras of the corresponding essential Lie invariance algebra
up to inner automorphisms and up to the action of the essential point-symmetry group.
This allowed us to classify Lie reductions and Lie invariant solutions of the equation under consideration.
We also discuss the generation of its solutions using point and linear generalized symmetries
and carry out its peculiar generalized reductions.
As a result, we construct wide families of its solutions
parameterized by an arbitrary finite number of arbitrary solutions of the (1+1)-dimensional linear heat equation
or one or two arbitrary solutions of (1+1)-dimensional linear heat equations with inverse square potentials.
}\par\vspace{4mm}

\noprint{
Keywords:
(1+2)-dimensional ultraparabolic linear Kolmogorov backward equations;
point-symmetry group;
Lie reductions;
generalized symmetry;
exact solutions;
classification of subalgebras

MSC: 35B06, 35K57, 35C05, 35A30
35-XX   Partial differential equations
  35A30   Geometric theory, characteristics, transformations [See also 58J70, 58J72]
  35B06   Symmetries, invariants, etc.
 35Cxx  Representations of solutions
  35C05   Solutions in closed form
  35C06   Self-similar solutions
  35C07   Traveling wave solutions
  35C08   Soliton solutions
  35K57   Reaction-diffusion equations
}

\section{Introduction}\label{sec:Introduction}

Fokker--Planck and Kolmogorov equations are widely used to model various real-world phenomena
in the broad range of sciences,
including various branches of physics, biology and economics~\cite{gard1985A,risk1989A}.
For instance, Kolmogorov~\cite{kolm1934a} suggested to describe the Brownian motion
with ultraparabolic Fokker--Planck equations of the general form
\begin{gather*}
u_t+xu_y=(f(t,x,y)u)_{xx}+(g(t,x,y)u)_x,
\end{gather*}
where~$f$ and~$g$ are arbitrary smooth functions depending on $(t,x,y)$.
He also constructed (up to the correcting multiplier~$\frac14$~\cite{kova2023a})
the fundamental solutions of such equations with constant parameter functions~$f$ and~$g$.
The equations of the above form with $f=\const$ and $g=0$
appeared later for modeling processes of evolution of cell populations~\cite[Eq.~30]{rote83a}
with constant diffusivity.
The Fokker--Planck equations are also appreciated in modeling lasers and
other nonlinear systems far from thermal equilibrium~\cite[Chapter~12]{risk1989A}.
Allowing the diffusivity to be a power function, one derives a slight generalization,
\begin{gather}\label{eq:classFPbeta}
u_t+xu_y=|x|^\beta u_{xx},
\end{gather}
where $\beta$ is a real parameter.
At the same time, these equations constitute a family of interesting objects
for the study within the framework of group analysis of differential equations.
We turn the collection of equations~\eqref{eq:classFPbeta}
into a class $\mathcal F$ by declaring~$\beta$ to be its only arbitrary element,
i.e., the tuple of arbitrary elements of $\mathcal F$ is $\theta=(\beta)$,
where~$\beta$ runs through the solution set of the system of equations $\beta_t=\beta_x=\beta_y=0$.
See~\cite{bihl2012b,bihlo2016a,popo2010a,vane2020a} and references therein for the definition
and transformational properties of classes of differential equations.
Symmetry properties of the class~$\mathcal F$ were partially considered in~\cite{zhan2020a}.
Each equation of the form~\eqref{eq:classFPbeta} belongs to the general class~$\bar{\mathcal F}$
of $(1+2)$-dimensional ultraparabolic linear partial differential equations
\begin{gather*}
\begin{split}
&u_t+B(t,x,y)u_y=A^2(t,x,y)u_{xx}+A^1(t,x,y)u_x+A^0(t,x,y)u+C(t,x,y) \\
&\text{with}\quad A^2\neq 0,\quad B_x\neq 0.
\end{split}
\end{gather*}
The tuple of arbitrary elements of the class~$\bar{\mathcal F}$ is $\bar\theta:=(B,A^2,A^1,A^0,C)$,
where the components of $\bar\theta$ range through the solution set of the system of inequalities $A^2\ne0$ and $B_x\ne0$
with no restrictions on the other parameter-functions.
A partial preliminary group classification of the class $\bar{\mathcal F}$ was carried out in~\cite{davy2015a}.
Some subclasses of the class~$\bar{\mathcal F}$ and separate equations were considered within the Lie-symmetry framework
in \cite{gung2018b,kova2013a,kova2023a,spic1997a,spic1999a,spic2011a}.

\looseness=-1
The equations~$\mathcal F_0$ and~$\mathcal F_2$ of the form~\eqref{eq:classFPbeta} with $\beta=0$ and $\beta=2$
are distinguished by their symmetry properties
not only in the narrow class~$\mathcal F$ but also in the very wide class~$\bar{\mathcal F}$.
More specifically, the essential Lie invariance algebra of the equation~$\mathcal F_0$ is eight-dimensional,
which is the maximum of such dimensions among equations from the class~$\bar{\mathcal F}$.
Moreover, possessing the essential Lie invariance algebra of dimension eight
is an attribute of the equation~$\mathcal F_0$ that singles it out, up to the point equivalence,
within the class~$\bar{\mathcal F}$.
This is why the second equation in the class~$\mathcal F$
with eight-dimensional essential Lie invariance algebra, $\mathcal F_5$,
is necessarily equivalent to~$\mathcal F_0$ with respect to point transformations.
The essential Lie invariance algebra of the equation~$\mathcal F_2$ is five-dimensional and nonsolvable,
and modulo the point equivalence, it is the only equation from the class~$\bar{\mathcal F}$ with this property.
Within the class $\mathcal F$, there are only two equations with five-dimensional essential Lie invariance algebras,
$\mathcal F_2$ and~$\mathcal F_3$, and they are similar with respect to a point transformation.
Because of this, we refer to the equations~\eqref{eq:classFPbeta} with $\beta=0$ and $\beta=2$
as the {\it remarkable Fokker--Planck equation} and the {\it fine Kolmogorov backward equation}, respectively.
The equation~$\mathcal F_0$ is the canonical form of the ultraparabolic Fokker--Planck equations
with constant diffusion and drift coefficients~$f$ and~$g$,
whose study was initiated by Kolmogorov in~\cite{kolm1934a} as mentioned above.
Moreover, it is shown in~\cite{baru2001a} that
the Kolmogorov equations describing the price of the geometric and arithmetic average strike call option
are reduced by point transformations to Eqs.~(3.4) and~(4.5)~\cite{baru2001a},
which coincide with~$\mathcal F_0$ and~$\mathcal F_2$ up to alternating the sign of~$y$, respectively.

Due to its special form,
the remarkable Fokker--Planck equation can also be considered as a Kolmogorov backward equation.
Its extended symmetry analysis was carried out in~\cite{kova2023a},
where we computed its point-symmetry pseudogroup using the direct method and analyzed its structure,
which allowed us to list one- and two-dimensional subalgebras of its essential and maximal Lie invariance algebra
with a view toward the classification of the Lie reductions and finding Lie invariant solutions.
Moreover, we have carried out peculiar generalized reductions that allowed us to construct
three wide families of the solutions of the remarkable Fokker--Planck equation parameterized
by an arbitrary finite number of solutions of the linear (1+1)-dimensional heat equation.

The subject of the present paper is a similar analysis for the fine Kolmogorov backward equation,
\begin{gather}\label{eq:FineFP}
u_t+xu_y=x^2u_{xx}.
\end{gather}
The paper is organized as follows.
In Section~\ref{sec:FineFPMIA} we describe the structure and important properties
of the maximal and essential Lie invariance algebras~$\mathfrak g$ and~$\mathfrak g^{\rm ess}$
of the equation~\eqref{eq:FineFP}.
Using an advanced version of the direct method,
which involves the known explicit form of the (usual) equivalence pseudogroup of the class $\bar{\mathcal F}$
and its normalization in the usual sense,
in Section~\ref{sec:FineFPPointSymGroup} we compute the point-symmetry pseudogroup~$G$
of the equation~\eqref{eq:FineFP}.
We thoroughly analyze the structure of the pseudogroup~$G$ and modify its multiplication operation.
This allows us to single out its essential subgroup~$G^{\rm ess}$
and properly classify the discrete symmetries of~\eqref{eq:FineFP}.
Section~\ref{sec:Subalgebras} is devoted to the classifications of subalgebras of the algebra~$\mathfrak g^{\rm ess}$
up to the equivalences generated by the inner automorphisms of this algebra
and by the action 
of the Lie group~$G^{\rm ess}$ on~$\mathfrak g^{\rm ess}$ as its Lie algebra.
The latter classification forms a basis for exhaustively classifying
the Lie reductions of~\eqref{eq:FineFP} in Section~\ref{sec:LieReductions}.
As a result, we construct wide families of Lie invariant solutions therein,
two families in terms of the general solution of the (1+1)-dimensional linear~heat equation
and one family in terms of the general solution of the (1+1)-dimensional linear heat equations
with inverse square potentials, where the potential coefficient is an additional parameter.
These solutions are used in Section~\ref{sec:SolutionGeneration}
as seed solutions when generating new solutions of the equation~\eqref{eq:FineFP} from known ones
by acting with the point symmetries of this equation and with its Lie-symmetry operators.
To the best of our knowledge, this is the first systematic consideration of the latter action in the literature.
The subject of Section~\ref{sec:GenReductions} is given by
peculiar generalized reductions and their relation to generating solutions by Lie-symmetry operators.
One of the obtained families of solutions is parameterized
by an arbitrary finite number of arbitrary solutions of the (1+1)-dimensional linear heat equation.
The results of the paper are summarized in Section~\ref{sec:Conclusion}.

For readers' convenience,
the constructed exact solutions of the equation~\eqref{eq:FineFP} are marked by the bullet symbol~$\bullet$\,.

\section{Lie invariance algebra}\label{sec:FineFPMIA}

The maximal Lie invariance algebra~$\mathfrak g$ of~\eqref{eq:FineFP} is spanned by the vector fields
\begin{gather*}
\mathcal P^y=\p_y,\quad
\mathcal D  =x\p_x+y\p_y,\quad
\mathcal K  = 2xy\p_x+y^2\p_y-xu\p_u,
\\
\mathcal P^t=\p_t,\quad
\mathcal I  =u\p_u,\quad
\mathcal Z(f)=f(t,x,y)\p_x,
\end{gather*}
where the parameter-function $f$ of $(t,x,y)$ ranges through the solution set of~\eqref{eq:FineFP}.
We compute~$\mathfrak g$ as well as other invariance algebras in the present paper
using the packages {\sf DESOLV} \cite{carm2000a} and {\sf Jets} \cite{BaranMarvan} for {\sf Maple};
the latter package is based on algorithms developed in~\cite{marv2009a}.

The set $\mathfrak g^{\rm lin}:=\{\mathcal Z(f)\}$ is an infinite-dimensional abelian ideal
of the algebra~$\mathfrak g$ associated with the linear superposition of solutions.
The subalgebra $\mathfrak g^{\rm ess}$ is supplemental to the ideal $\mathfrak g^{\rm lin}$ in $\mathfrak g$,
\begin{gather}\label{eq:FineFPEssA}
\mathfrak g^{\rm ess}=\langle\mathcal P^y,\mathcal D,\mathcal K,\mathcal P^t,\mathcal I\rangle,
\end{gather}
therefore $\mathfrak g$ is a semidirect product $\mathfrak g^{\rm ess}\lsemioplus\mathfrak g^{\rm lin}$.

Up to the skew-symmetry of the Lie bracket,
the nontrivial commutation relations are
\begin{gather*}
[\mathcal P^y,\mathcal D]=\mathcal P^y,\quad
[\mathcal P^y,\mathcal K]=2\mathcal D,\quad
[\mathcal D,\mathcal K]  =\mathcal K.
\end{gather*}

The algebra $\mathfrak g^{\rm ess}$ is nonsolvable
since it contains a subalgebra $\mathfrak f=\langle\mathcal P^y,\mathcal D,\mathcal K\rangle$ that is isomorphic to the algebra ${\rm sl}(2,\mathbb R)$.
A complementary subalgebra to $\mathfrak f$ in $\mathfrak g^{\rm ess}$ is $\langle\mathcal P^t,\mathcal I\rangle$,
which is the center~$\mathfrak z$ of the algebra~$\mathfrak g^{\rm ess}$.
Hence $\mathfrak g^{\rm ess}=\mathfrak f\oplus\mathfrak z$.
Therefore, the entire algebra $\mathfrak g^{\rm ess}$ is isomorphic to the algebra ${\rm sl}(2,\mathbb R)\oplus2A_1$,
where $2A_1$ is an abelian two-dimensional Lie algebra, see~\cite{popo2003a} for notation.
Moreover, since ${\rm sl}_2(\mathbb R)\oplus A_1\simeq{\rm gl}(2,\mathbb R)$,
we have the decomposition $\mathfrak g^{\rm ess}=\mathfrak h\oplus\mathfrak p$,
where $\mathfrak h:=\langle\mathcal P^y,\mathcal D,\mathcal K,\mathcal I\rangle\simeq{\rm gl}(2,\mathbb R)$
and $\mathfrak p:=\langle\mathcal P^t\rangle\simeq A_1$.

The Iwasawa decomposition for the algebra ${\rm sl}(2,\mathbb R)$ gives another basis for $\mathfrak g^{\rm ess}$,
which is more convenient from many prospectives for analysis than the chosen one.
This basis is given by the tuple $(\mathcal Q^+,\mathcal D,\mathcal P^y,\mathcal P^t,\mathcal I)$,
where $\mathcal Q^\pm:=\mathcal P^y\pm\mathcal K$.

\section{Point symmetry pseudogroup}\label{sec:FineFPPointSymGroup}

The equation~\eqref{eq:FineFP} corresponds to the value $(x,x^2,0,0,0)=:\bar\theta_2$
of the arbitrary-element tuple~$\bar\theta=(B,A^2,A^1,A^0,C)$ of the class~$\bar{\mathcal F}$.
To find the point symmetry pseudogroup~$G$ of~\eqref{eq:FineFP},
we use the advanced version of the direct method based on
the known transformational properties of the class~$\bar{\mathcal F}$ that were presented in~\cite{kova2023a}.
(A similar approach was first applied in~\cite{bihlo2011b}
involving some results from~\cite{popo2012a}.
For the standard version of the direct method, see~\cite{king1998a}.)
In particular, the class~$\bar{\mathcal F}$ is normalized in the usual sense, i.e.,
its equivalence groupoid coincides with the action groupoid of
its usual equivalence pseudogroup~$G^\sim_{\bar{\mathcal F}}$,
which consists of the point transformations with the components~\cite[Theorem~1]{kova2023a}
\begin{subequations}\label{eq:EquivalenceGroupFPsuperClass}
\begin{gather}
\tilde t=T(t,y),
\quad
\tilde x=X(t,x,y),
\quad
\tilde y=Y(t,y),
\quad
\tilde u=U^1(t,x,y)u + U^0(t,x,y),
\label{eq:ClassFbarTransformationPart}\\[.5ex]
\tilde A^0=\frac{-U^1}{T_t+BT_y}E\frac1{U^1},
\quad
\tilde A^1=A^1\frac{X_x}{T_t+BT_y}-\frac{X_t+BX_y}{T_t+BT_y}+A^2\frac{X_{xx}-2X_xU^1_x/U^1}{T_t+BT_y},
\label{eq:A^0A^1Transformation}\\[.5ex]
\tilde A^2=A^2\frac{X_x^2}{T_t+BT_y},
\quad
\tilde B=\frac{Y_t+BY_y}{T_t+BT_y},
\quad
\tilde C=\frac{U^1}{T_t+BT_y}\left(C-E\frac{U^0}{U^1}\right).
\label{eq:A^2BCTransformation}
\end{gather}
\end{subequations}
Here $T$, $X$, $Y$, $U^1$ and $U^0$ are arbitrary smooth functions of their arguments
with the nondegeneracy condition $(T_tY_y-T_yY_t)X_xU^1\neq0$, and $E:=\p_t+B\p_y-A^2\p_{xx}-A^1\p_x-A^0$.
The normalization of~$\bar{\mathcal F}$ implies that the latter groupoid necessarily contains
the vertex group~$\mathcal G^\sim_{\bar{\mathcal F}}(\bar\theta_2,\bar\theta_2)$ of~\eqref{eq:FineFP},
which is the set of all admissible transformations with $\bar\theta_2$ as both their target and source,
$\mathcal G^\sim_{\bar{\mathcal F}}(\bar\theta_2,\bar\theta_2)=\{(\bar\theta_2,\Phi,\theta_2)\mid \Phi\in G\}$.
These arguments allow us to use the description of~$G^\sim_{\bar{\mathcal F}}$ for computing~$G$.

\begin{theorem}\label{thm:FineFPSymGroup}
The point symmetry pseudogroup~$G$ of the fine Kolmogorov backward equation~\eqref{eq:FineFP}
consists of the point transformations of the form
\begin{gather}\label{eq:FineFPSymGroup}
\begin{split}
&\tilde t=t+\lambda,
\quad
\tilde x=\frac{\alpha\delta-\beta\gamma}{(\gamma y+\delta)^2}x,
\quad
\tilde y=\frac{\alpha y+\beta}{\gamma y+\delta},
\\[1ex]
&\tilde u=\sigma\exp\left(
-\frac{\gamma}{\gamma y+\delta}x
\right)
\big(u+f(t,x,y)\big),
\end{split}
\end{gather}
where
$\alpha$, $\beta$, $\gamma$ and $\delta$ are arbitrary constants with $\alpha\delta-\beta\gamma=\pm1$
that are defined up to simultaneously alternating their signs;
$\lambda$ and $\sigma$ are arbitrary constants with $\sigma\ne0$,
and $f$ is an arbitrary solution of the equation~\eqref{eq:FineFP}.	
\end{theorem}

\begin{proof}
We integrate the system~\eqref{eq:EquivalenceGroupFPsuperClass},
where both source and target arbitrary-element tuples are equal to~$\bar\theta_2$,
with respect to the parameter-functions $T$, $X$, $Y$ and $U$.
The first two equations from~\eqref{eq:A^2BCTransformation} take the form
\[
X=\dfrac{Y_t +xY_y}{T_t +xT_y},\quad
X^2(T_t+xT_y)=(xX_x)^2.
\]
In view of the first of these equations, the second one reduces to the equation
\begin{gather*}
(Y_t+xY_y)^2(T_t+xT_y)^3-x^2(Y_yT_t-Y_tT_y)^2=0,
\end{gather*}
whose left-hand side is a polynomial of~$x$ with coefficients depending on~$t$ and~$y$.
Collecting the coefficients of~$x^5$ and~$x^0$, we obtain $T_y^3Y_y^2=T_t^3Y_t^2=0$,
i.e., in view of the nondegeneracy condition, either $T_t=Y_y=0$ or $T_y=Y_t=0$.
In the first case, collecting of the coefficients of~$x^3$
results in the equation $T_yY_t=0$,
which contradicts the nondegeneracy condition.
Hence the only possible case is $T_y=Y_t=0$,
where collecting the coefficients of~$x^2$ leads to the equation $T_t^2Y_y^2(T_t-1)=0$.
In view of the inequality $T_tY_y\ne0$, we derive the equation $T_t=1$,
which integrates to $T=t+\lambda$, where $\lambda$ is an arbitrary constant.

Moreover, $T_y=Y_t=0$ together with the first equation from~\eqref{eq:A^2BCTransformation}
give us that $X=xY_y$.
Therefore, $X_t=X_{xx}=0$.
Then the second equation in~\eqref{eq:A^0A^1Transformation} simplifies to $U^1_x/U^1=-Y_{yy}/(2Y_y)$,
and thus it integrates to
\[
U^1:=V(t,y)\exp\left(-\frac{Y_{yy}}{2Y_y}x\right),
\]
where $V$ is a nonvanishing smooth function of~$(t,y)$, the explicit expression for which will be derived below.
Substituting the obtained expression for~$U^1$ into the first equation in~\eqref{eq:A^0A^1Transformation}
results in the equation
\[
x^2\left(\frac{U^1_x}{U^1}\right)^2+\frac{U^1_t+xU^1_y}{U^1}=0,
\]
whose left-hand side is a quadratic polynomial in~$x$.
Collecting the coefficients of~$x^2$, we get the equation $Y_{yyy}Y_y-\frac32Y_{yy}^2=0$,
meaning that the Schwarzian derivative of $Y$ is zero.
Therefore, $Y$ is a linear fractional function of $y$, $Y=(\alpha y+\beta)/(\gamma y+\delta)$.
Since the constants $\alpha$, $\beta$, $\gamma$ and $\delta$ are defined up to a nonzero constant multiplier,
we can assume that $\alpha\delta-\beta\gamma=\pm1$.
Collecting the coefficients of $x^1$ and $x^0$, we obtain the equations $V_y=V_t=0$.
Hence $V=:\sigma\in\mathbb R\setminus\{0\}$.

Finally, the third equation in~\eqref{eq:A^2BCTransformation} takes the form
\[
\left(\dfrac{U^0}{U^1}\right)_t+x\left(\dfrac{U^0}{U^1}\right)_y=x^2\left(\dfrac{U^0}{U^1}\right)_{xx},
\]
and thus $U^0=U^1f$, where $f=f(t,x,y)$ is an arbitrary solution of~\eqref{eq:FineFP}.

Clearly, the set of transformations of the form~\eqref{eq:FineFPSymGroup} is closed under the composition,
therefore~$G$ is a pseudogroup.
Moreover, it is a Lie pseudogroup.
\end{proof}

The natural maximal domain of a point transformation of the form~\eqref{eq:FineFPSymGroup}
coincides with the relative complement of~$M_{\gamma\delta}:=\{(t,x,y,u)\in\mathbb R^4\mid\gamma y+\delta=0\}$
with respect to the set ${\mathop{\rm dom}f\times\mathbb R_u}$.
Since $(\gamma,\delta)\ne(0,0)$, the set~$M_{\gamma\delta}$ is
either hyperplane in $\mathbb R^4_{t,x,y,u}$ defined by $t=-\delta/\gamma$ if $\gamma\ne0$
or the empty set otherwise.
To properly interpret the algebraic structure on the set of such transformations,
we consider them on natural maximal domains together with all their restrictions
as partial functions defined on the entire space $\mathbb R^4_{t,x,y,u}$
rather than functions defined on subsets of this space.
Considering the composition of these partial functions as the operation on~$G$,
we endow~$G$ with the \textit{pseudogroup} structure in the sense of Definition~\ref{def:Pseudogroup},
which is a particular case of the more general algebraic structure called \textit{inverse monoid},
see Section~\ref{sec:Pseudogroups}.

Since the $y$-component of any element of~$G$ is a linear fractional function in $y$,
we follow the papers \cite{kova2023a,kova2023b} and take
the \emph{modified transformation composition} $\Phi_1\circ^{\rm m}\Phi_2$ of $\Phi_1,\Phi_2\in G$ as the operation in $G$.
We first consider transformations~$\Phi_1$ and~$\Phi_2$ of the form~\eqref{eq:FineFPSymGroup}
having maximal domains for this form.
The domain of their composition $\Phi_1\circ\Phi_2=:\tilde\Phi$
is standardly defined as the preimage of the domain of $\Phi_1$ under $\Phi_2$,
\[
\mathop{\rm dom}\tilde\Phi=\Phi_2^{-1}(\mathop{\rm dom}\Phi_1)
=(\mathop{\rm dom}\tilde f\times\mathbb R_u)\setminus (M_{\gamma_2\delta_2}\cup M_{\tilde\gamma\tilde\delta}),
\]
where
$\tilde\gamma=\gamma_1\alpha_2+\delta_1\gamma_2$,
$\tilde\delta=\gamma_1\beta _2+\delta_1\delta_2$,
$\mathop{\rm dom}\tilde f=\big((\pi_*\Phi_2)^{-1}\mathop{\rm dom}f^1\big)\cap\mathop{\rm dom}f^2$
with the natural projection \smash{$\pi\colon\mathbb R^4_{t,x,y,u}\twoheadrightarrow\mathbb R^3_{t,x,y}$},
and the parameters with indices~1 and~2 and tildes correspond~$\Phi_1$, $\Phi_2$ and~$\tilde\Phi$, respectively.
Modifying the composition $\Phi_1\circ\Phi_2$ to $\Phi_1\circ^{\rm m}\Phi_2$,
we continuously extend $\Phi_1\circ\Phi_2$ to the set
\[\smash{\mathop{\rm dom}\nolimits^{\rm m}\tilde\Phi:=(\mathop{\rm dom}\tilde f\times\mathbb R_u)\setminus M_{\tilde\gamma\tilde\delta}},\]
i.e.,
the modified composition $\Phi_1\circ^{\rm m}\Phi_2$ is the transformation of the form~\eqref{eq:FineFPSymGroup}
with the same parameters as in $\Phi_1\circ\Phi_2$ and with the natural domain,
\smash{$\mathop{\rm dom}(\Phi_1\circ^{\rm m}\Phi_2)=\mathop{\rm dom}^{\rm m}\tilde\Phi$}.
In other words, we redefine $\Phi_1\circ\Phi_2$ on the set
\smash{$(\mathop{\rm dom}\tilde f\times\mathbb R_u)\cap M_{\gamma_2\delta_2}$} if $\gamma_1\gamma_2\ne0$;
otherwise $\mathop{\rm dom}\nolimits^{\rm m}\tilde\Phi=\mathop{\rm dom}\tilde\Phi$
and the extension is trivial.

If transformations $\Phi_1$ and $\Phi_2$ from $G$ are defined on open subsets of maximal domains
of the corresponding forms~\eqref{eq:FineFPSymGroup},
then we also redefine the domain of their composition in a natural way.
More specifically, the set $\mathop{\rm dom}(\Phi_1\circ\Phi_2)=\Phi_2^{-1}(\mathop{\rm dom}\Phi_1)$ is contained in
$(\mathop{\rm dom}\tilde f\times\mathbb R_u)\setminus (M_{\gamma_2\delta_2}\cup M_{\tilde\gamma\tilde\delta})$
and if $\gamma_1\gamma_2\ne0$
as the modified composition $\Phi_1\circ^{\rm m}\Phi_2$
we take the continuous extension of $\Phi_1\circ\Phi_2$ to the~set
\[
{\rm int}\Big(\Phi_2^{-1}(\mathop{\rm dom}\Phi_1)\cup
\big(\mathop{\rm cl}\big(\Phi_2^{-1}(\mathop{\rm dom}\Phi_1)\big)\cap M_{\gamma_2\delta_2}\cap (\mathop{\rm dom}\tilde f\times\mathbb R_u\big)\big)
\Big).
\]
In other words, we complete the set $\mathop{\rm dom}(\Phi_1\circ\Phi_2)$ by its limit points from~$M_{\gamma_2\delta_2}$
that belong to the interior of the completed set.
If $\gamma_1\gamma_2=0$ we simply set  $\Phi_1\circ^{\rm m}\Phi_2=\Phi_1\circ\Phi_2$,
thus $\mathop{\rm dom}\Phi_1\circ^{\rm m}\Phi_2=\Phi_2^{-1}(\mathop{\rm dom}\Phi_1)$.
By this, it is clear that $\Phi_1\circ^m\Phi_2\in G$ for arbitrary $\Phi_1,\Phi_2\in G$,
thus the set~$G$ with the modified composition $\circ^{\rm m}$ is indeed a pseudogroup,
and, moreover, a Lie pseudogroup.

This allows us to appropriately describe the structure of the pseudogroup~$G$.
The set $G^{\rm lin}$ constituted by the point transformations of the form
\[
\mathscr Z(f)\colon\quad
\tilde t=t,\quad
\tilde x=x,\quad
\tilde y=y,\quad
\tilde u=u+f(t,x,y),
\]
where the function~$f$ of $(t,x,y)$ ranges through the solution set of~\eqref{eq:FineFP},
is a normal pseudosubgroup of~$G$;
moreover, the pseudogroup~$G$ splits over $G^{\rm lin}$, $G=G^{\rm ess}\ltimes G^{\rm lin}$,
see Section~\ref{sec:Pseudogroups} for all the necessary definitions.

Here $G^{\rm ess}$ is the subgroup of $G$ constituted by the transformations of the form~\eqref{eq:FineFPSymGroup}
with $f=0$ and the natural domain of definition,
and thus it is a five-dimensional Lie group.
We call $G^{\rm ess}$ the {\it essential point symmetry group} of the equation~\eqref{eq:FineFP}.

The pushforward $\pi_*G$ of the pseudogroup~$G$
under the natural projection ${\pi\colon\mathbb R^4_{t,x,y,u}\twoheadrightarrow\mathbb R_y}$
coincides with the pseudogroup of M\"obius transformations of the real line.
At the same time, M\"obius transformations with their natural domains
constitute a group with respect to the modified composition as the group operation~\cite{laws1998a}.
In this sense, the above consideration extends the approach described in~\cite{laws1998a}
to more general transformation groups like~$G^{\rm ess}$ and analogous groups in~\cite{kova2023a,kova2023b}.

The group $G^{\rm ess}$ itself is the direct product, $G^{\rm ess}=F\times Z$,
of its normal subgroups~$F$ and~$Z$
that are singled out by the constraints $\lambda=0$, $\sigma=1$ and $\alpha=\delta=1$, $\beta=\gamma=0$
and are isomorphic, as Lie groups, to the groups ${\rm PSL}^\pm(2,\mathbb R)$ and $(\mathbb R^2,+)\times\mathbb Z_2$, respectively.
Here $(\mathbb R^2,+)$ is the real two-dimensional connected torsion-free abelian Lie group.
Recall that ${\rm PSL}^\pm(2,\mathbb R)$ is by definition the quotient group
${\rm SL}^{\pm}(2,\mathbb R)/{\rm Z(SL^\pm(2,\mathbb R))}$,
where ${\rm Z}({\rm SL}^\pm(2,\mathbb R))=\{{\pm\rm diag}(1,1)\}$
is the center of the group ${\rm SL}^\pm(2,\mathbb R)$,
which coincides with the center of the group ${\rm SL}(2,\mathbb R)$.
The above isomorphisms are established by the correspondences
\begin{gather*}
\varrho_1=\pm(\alpha,\beta,\gamma,\delta)_{\alpha\delta-\beta\gamma=\pm1}\mapsto
\begin{pmatrix}
	\alpha & \beta \\
	\gamma  & \delta
\end{pmatrix}{\rm Z(SL(2,\mathbb R))},
\\
(\lambda,\sigma)_{\sigma>0}\mapsto(\lambda, \ln|\sigma|,\bar 0)
\quad\mbox{and}\quad
(\lambda,\sigma)_{\sigma<0}\mapsto(\lambda, \ln|\sigma|,\bar 1),
\end{gather*}
where $\bar0:=2\mathbb Z$ and $\bar1:=1+2\mathbb Z$ constitute the group $\mathbb Z_2:=\mathbb Z/2\mathbb Z$.
Summing up, we have that
\begin{gather*}
G^{\rm ess}\simeq{\rm PSL}^\pm(2,\mathbb R)\times(\mathbb R^2,+)\times\mathbb Z_2.
\end{gather*}
This decomposition corresponds to the Lie algebra decomposition
$\mathfrak g^{\rm ess}=\mathfrak f\oplus\mathfrak r\simeq{\rm sl}(2,\mathbb R)\oplus2A_1$.

There is another factorization of the group $G^{\rm ess}$.
It contains a normal subgroup $H$,
which is singled out by the constraints $\lambda=0$ and $\sigma>0$
and is isomorphic to the Lie group ${\rm GL}(2,\mathbb R)$.
The group-complement to~$H$ in the group~$G$ is also a normal subgroup~$P$ of~$G$.
It is singled out by the constraints
$\sigma=\pm1$, $\alpha=\delta=1$ and $\beta=\gamma=0$
and is isomorphic to $(\mathbb R,+)\times\mathbb Z_2$.
The isomorphisms between $H$ and ${\rm GL}(2,\mathbb R)$ and
between $P$ and $(\mathbb R,+)\times\mathbb Z_2$ are established by the correspondences
\begin{gather*}
\varrho_2=(\alpha,\beta,\gamma,\delta,\sigma)_{\alpha\delta-\beta\gamma=\pm1,\sigma>0}\mapsto
\sigma
\begin{pmatrix}
\alpha & \beta \\
\gamma  & \delta
\end{pmatrix},
\quad
(\lambda,\sigma)_{\sigma=\pm1}\mapsto
\begin{cases}
(\lambda,\bar 1)\text{ if }\sigma=-1,\\
(\lambda,\bar 0)\text{ if }\sigma= 1,
\end{cases}
\end{gather*}
In other words, the group $G^{\rm ess}$ admits the following decomposition:
\[G^{\rm ess}=H\times P\simeq{\rm GL}(2,\mathbb R)\times(\mathbb R,+)\times\mathbb Z_2.\]
The counterpart of this isomorphism in terms of Lie algebras
is the Lie algebra decomposition $\mathfrak g^{\rm ess}=\mathfrak h\oplus\mathfrak p$,
given in Section~\ref{sec:FineFPMIA}.

Note that the former and the latter decompositions of $G^{\rm ess}$ are consistent.
Recall that the group ${\rm GL}(2,\mathbb R)$ is isomorphic to ${\rm PGL}(2,\mathbb R)\times(\mathbb R,+)$,
where ${\rm PGL}(2,\mathbb R)$ is the rank-two real projective linear group
${\rm PGL}(2,\mathbb R):={\rm GL}(2,\mathbb R)/{\rm Z}({\rm GL}(2,\mathbb R))$.
Here the center ${\rm Z}({\rm GL}(2,\mathbb R))$ of the group ${\rm GL}(2,\mathbb R)$ coincides with the set of scalar matrices,
${\rm Z}({\rm GL}(2,\mathbb R))=\{\lambda\mathop{\rm diag}(1,1)\mid\lambda\in\mathbb R^\times\}$,
and as a Lie group it is isomorphic to the group $(\mathbb R,+)\times\mathbb Z_2$.
It is clear that ${\rm PGL}(2,\mathbb R)\simeq{\rm PSL}^\pm(2,\mathbb R)$,
therefore ${\rm PSL}^\pm(2,\mathbb R)\times(\mathbb R,+)\simeq{\rm GL}(2,\mathbb R)$.
\noprint{
The isomorphism ${\rm PGL}(2,\mathbb R)\times(\mathbb R,+)\simeq{\rm GL}(2,\mathbb R)$ is established by the correspondence
\begin{gather*}
\left(\begin{pmatrix}
a & b\\
c & d
\end{pmatrix}
{\rm Z(GL}(2,\mathbb R)),\sigma\right)
\mapsto
\frac{\rm e^\sigma}{\sqrt{|ad-bc|}}
\begin{pmatrix}
a & b\\
c & d
\end{pmatrix}
\end{gather*}
}

The one-parameter subgroups of the group~$G^{\rm ess}$
that are generated by the chosen basis elements of the algebra~$\mathfrak g^{\rm ess}$, see~\eqref{eq:FineFPEssA},
are of the form
\[\arraycolsep=0ex
\begin{array}{lllll}
\mathscr P^y(\epsilon)\colon  & \tilde t=t,\quad    & \tilde x=x,\quad                        & \tilde y=y+\epsilon,\quad             & \tilde u=u,\\[.8ex]
\mathscr D  (\epsilon)\colon  & \tilde t=t,\quad    & \tilde x={\rm e}^\epsilon x,\quad       & \tilde y={\rm e}^\epsilon y,\quad     & \tilde u=u,\\
\mathscr K(\epsilon)  \colon  & \tilde t=t,\quad
                              & \tilde x=\dfrac x{(1-\epsilon y)^2},\quad
                              & \tilde y=\dfrac y{1-\epsilon y},\quad
                              & \tilde u=u{\rm e}^{\frac{\epsilon x}{\epsilon y-1}},
\\[1ex]
\mathscr P^t(\epsilon)\colon  & \tilde t=t+\epsilon,\quad & \tilde x=x,\quad  & \tilde y=y,\quad  & \tilde u=u,\\[.8ex]
\mathscr I  (\epsilon)\colon  & \tilde t=t,\quad          & \tilde x=x,\quad  & \tilde y=y,\quad  & \tilde u={\rm e}^\epsilon u
\end{array}
\]
with constant~$\epsilon$ playing the role of group parameter.
The one-parameter subgroup generated by the vector field $\mathcal Q^+=\mathcal P^t+\mathcal K$,
which stems from the Iwasawa decomposition for the algebra~${\rm sl}(2,\mathbb R)$,
consists of the transformations
\begin{gather*}
\mathscr Q^+(\epsilon)\colon\
\tilde t=t,\ \
\tilde x=\frac x{(\cos\epsilon-y\sin\epsilon)^2},\ \
\tilde y=\frac{\sin\epsilon+y\cos\epsilon}{\cos\epsilon-y\sin\epsilon},\ \
\tilde u=u\exp\left(\frac{x\sin\epsilon}{y\sin\epsilon-\cos\epsilon}\right)
\end{gather*}
parameterized by $\epsilon\in\mathbb R$ as well.
This subgroup is distinguished among the others by its peculiar action on the Lie algebra $\mathfrak g^{\rm ess}$.

Corollary~\ref{cor:PSL+-} from Section~\ref{sec:StructureOfSLpm} implies that the subgroup $F$
is the semidirect product of its subgroups $F_{\rm id}$ and $F_{\rm d}$, $F=F_{\rm id}\rtimes F_{\rm d}$,
where the former is singled out in $F$ by the constraints $\alpha\delta-\beta\gamma=1$,
and the latter is generated by the transformation $(\tilde t,\tilde x,\tilde y,\tilde u)=(t,-x,-y,u)$.
It is also clear that the group $Z$ is the direct product of its subgroups $Z_{\rm id}$ and $Z_{\rm d}$,
where $Z_{\rm id}$ is singled out in $Z$ by $\sigma>0$,
and $Z_{\rm d}$ is generated by the transformation $(\tilde t,\tilde x,\tilde y,\tilde u)=(t,x,y,-u)$.
Moreover, due to the isomorphisms $F_{\rm id}\simeq{\rm PSL}(2,\mathbb R)$ and $Z_{\rm id}\simeq(\mathbb R^2,+)$
we have that the groups $F_{\rm id}$ and $Z_{\rm id}$ are connected.

\begin{corollary}
The center of the group $G^{\rm ess}$ coincides with the normal subgroup $Z$, ${\rm Z}(G^{\rm ess})=Z$.
\end{corollary}

\begin{proof}
Obviously $Z\subset{\rm Z}(G^{\rm ess})$.
Since $G^{\rm ess}/Z\simeq{\rm PSL}^\pm(2,\mathbb R)$,
then ${\rm Z}(G^{\rm ess}/Z)\simeq{\rm Z}({\rm PSL}^\pm(2,\mathbb R))=\{\pm{\rm diag}(1,1)\}$,
which is the identity element in the group ${\rm PSL}^\pm(2,\mathbb R)$.
This means that the group ${\rm Z}(G^{\rm ess}/Z)$ is trivial.
Observing that ${\rm Z}(G^{\rm ess})/Z\subset{\rm Z}(G^{\rm ess}/Z)$ we deduce ${\rm Z}(G^{\rm ess})\subset Z$.
\end{proof}

\begin{corollary}\label{prop:GessModGessId}
The quotient group of $G^{\rm ess}$ by its identity component $G^{\rm ess}_{\rm id}$ is isomorphic to the Klein four-group $\mathbb Z_2^2$.
\end{corollary}

\begin{proposition}\label{pro:FineFPDiscrSyms}
A complete list of discrete point symmetry transformations of the fine Kolmogorov backward equation~\eqref{eq:FineFP}
that are independent up to combining with each other and with continuous point symmetry transformations of this equation
is exhausted by the involution~$\mathscr I'$ alternating the sign of~$u$,
the involution~$\mathscr J'$ simultaneously alternating the signs of~$x$ and $y$,
\[\mathscr I':=(t,x,y,u)\mapsto(t,x,y,-u),
\quad
\mathscr J':=(t,x,y,u)\mapsto(t,-x,-y,u),\]
and their composition.
Thus, the group of discrete symmetries of~\eqref{eq:FineFP} is isomorphic to the Klein four-group~$\mathbb Z_2^2$.
\end{proposition}

\begin{proof}
Clearly, the entire pseudosubgroup~$G^{\rm lin}$ is contained in the connected component of the identity transformation in~$G$.
Since both $F_{\rm id}$ and $Z_{\rm id}$ are connected, then without loss of generality we can assume that the discrete symmetry transformations
constitute the group $F_{\rm d}\times Z_{\rm d}$,
and as the generators of it the transformations $\mathcal I'$ and $\mathcal J'$ can be chosen.
\end{proof}

\section{Classification of subalgebras}\label{sec:Subalgebras}

Subalgebras of a Lie algebra~$\mathfrak g$ can be classified up to various equivalence relations.
However, from the algebraic perspective, the most interesting and useful is
the classification modulo the action of the inner automorphism group~${\rm Inn}(\mathfrak g)$ of~$\mathfrak g$.
Moreover, when $\mathfrak g$ is a real (resp.\ complex) finite-dimensional Lie algebra,
this action coincides with the adjoint action of the corresponding real (resp.\ complex) simply connected Lie group
on~$\mathfrak g$.
Another natural equivalence for the subalgebra classification
is generated by the action of the entire automorphism group~${\rm Aut}(\mathfrak g)$ of~$\mathfrak g$.
Note that it often happens that there is no big difference between the above two equivalences
or, in more precise algebraic terms,
${\rm Inn}(\mathfrak g)={\rm Aut}(\mathfrak g)$
or at least the index of~${\rm Inn}(\mathfrak g)$ as a subgroup of~${\rm Aut}(\mathfrak g)$ is finite and small.

In group analysis of differential equations,
the classification of subalgebras of maximal (resp.\ essential) Lie invariance algebra $\mathfrak g$
of a system of differential equations~$\mathcal L$ is required in the course of classifying its Lie reductions
for the purpose of constructing group-invariant solutions.
Since these solutions are naturally considered modulo their similarity,
i.e., the equivalence generated by the action of the entire (resp.\ essential) point symmetry group~$G$
on solution set of the system~$\mathcal L$,
the most relevant subalgebra classification here is that up to the $G$-equivalence of subalgebras of~$\mathfrak g$.
If the group $G$ is not connected, this classification may differ from the classification
modulo the ${\rm Inn}(\mathfrak g)$-equivalence.

For these reasons, we carry out the classification of subalgebras of the algebra $\mathfrak g^{\rm ess}$
up to the $G^{\rm ess}_{\rm id}$-equivalence,
which coincides with the ${\rm Inn}(\mathfrak g^{\rm ess})$-equivalence,
and then construct a complete list of $G^{\rm ess}$-inequivalent subalgebras of~$\mathfrak g^{\rm ess}$.
Since the group $G^{\rm ess}$ is generated by continuous symmetries~$\mathscr P^y(\epsilon)$,
$\mathscr D(\epsilon)$, $\mathscr K(\epsilon)$, $\mathscr P^t(\epsilon)$ and $\mathscr I(\epsilon)$,
where $\epsilon\in\mathbb R$,
together with the discrete symmetries~$\mathscr I'$ and $\mathscr J'$,
the action of $G^{\rm ess}$ on $\mathfrak g^{\rm ess}$ is completely defined
by the actions of these generators on the chosen basis elements of $\mathfrak g^{\rm ess}$.
The only nonidentity actions among them are the following:
\begin{gather*}\arraycolsep=0ex
\begin{array}{lll}
\mathscr P^y(\epsilon)_*\mathcal D  =\mathcal D-2\epsilon\mathcal P^y,&
\mathscr K(\epsilon)_*\mathcal D  =\mathcal D+2\epsilon\mathcal K,&
\mathscr D(\epsilon)_*\mathcal P^y=e^{2\epsilon}\mathcal P^y,
\\[.75ex]
\mathscr P^y(\epsilon)_*\mathcal K  =\mathcal K-\epsilon\mathcal D+\epsilon^2\mathcal P^y,\qquad&
\mathscr K(\epsilon)_*\mathcal P^y=\mathcal P^y+\epsilon\mathcal D+\epsilon^2\mathcal K,\qquad&
\mathscr D(\epsilon)_*\mathcal K  =e^{-2\epsilon}\mathcal K,
\end{array}
\\
\arraycolsep=0ex
\begin{array}{l}
\mathscr J'_*\mathcal P^y=-\mathcal P^y,\qquad
\mathscr J'_*\mathcal K  =-\mathcal K.
\end{array}
\end{gather*}
It is also useful to consider the action of the one-parameter subgroup $\{\mathscr Q^+(\epsilon)\mid\epsilon\in\mathbb R\}$
on $\mathfrak g^{\rm ess}$,
\begin{gather*}
\mathscr Q^+(\epsilon)_*\mathcal P^y=\mathcal K\sin^2\epsilon+\mathcal D\sin2\epsilon+\mathcal P^y\cos^2\epsilon,
\\
\mathscr Q^+(\epsilon)_*\mathcal D  =\tfrac12\mathcal K\sin2\epsilon+\mathcal D\cos2\epsilon-\tfrac12\mathcal P^y\sin2\epsilon,
\\
\mathscr Q^+(\epsilon)_*\mathcal K  =\mathcal K\cos^2\epsilon-\mathcal D\sin2\epsilon+\mathcal P^y\sin^2\epsilon.
\end{gather*}

\begin{lemma}\label{lem:IneqSubalgebrasUpToInn}
A complete list of $G^{\rm ess}_{\rm id}$-inequivalent proper subalgebras of the algebra $\mathfrak g^{\rm ess}$
is exhausted by the following subalgebras, where $\delta\in\{-1,0,1\}$, $\varepsilon=\pm1$, $\mu,\mu'\in\mathbb R$,
$\nu\geqslant0$ and $\nu'>0$:
\begin{gather*}
{\rm 1D}\colon\quad
\mathfrak s_{1.1}^\mu               =\langle\mathcal P^t+\mu\mathcal I\rangle,\quad
\mathfrak s_{1.2}                   =\langle\mathcal I\rangle,\quad
\mathfrak s_{1.3}^{\varepsilon,\mu} =\langle\mathcal P^y+\varepsilon\mathcal P^t+\mu\mathcal I\rangle,\quad
\mathfrak s_{1.4}^\delta            =\langle\mathcal P^y+\delta\mathcal I\rangle,
\\
\hphantom{{\rm 1D}\colon\quad}
\smash{\mathfrak s_{1.5}^{\nu',\mu}}=\langle\mathcal D+\nu'\mathcal P^t+\mu\mathcal I\rangle,\quad
\mathfrak s_{1.6}^{\nu}             =\langle\mathcal D+\nu\mathcal I\rangle,\quad
\smash{\mathfrak s_{1.7}^{\mu,\mu'}}=\langle\mathcal P^y+\mathcal K+\mu\mathcal P^t+\mu'\mathcal I\rangle,
\\[1ex]
{\rm 2D}\colon\quad
\mathfrak s_{2.1}               =\langle\mathcal P^t,\mathcal I\rangle,\quad
\mathfrak s_{2.2}^{\delta,\mu}  =\langle\mathcal P^y+\delta\mathcal I,\mathcal P^t+\mu\mathcal I\rangle,\quad
\mathfrak s_{2.3}^\delta        =\langle\mathcal P^y+\delta\mathcal P^t,\mathcal I\rangle,
\\
\hphantom{{\rm 2D}\colon\quad}
\mathfrak s_{2.4}^{\nu,\mu} =\langle\mathcal D+\nu\mathcal I,\mathcal P^t+\mu\mathcal I\rangle,\quad
\mathfrak s_{2.5}^{\nu}     =\langle\mathcal D+\nu\mathcal P^t,\mathcal I\rangle,
\\
\hphantom{\rm 2D\colon}\quad
\smash{\mathfrak s_{2.6}^{\mu,\mu'}}=\langle\mathcal P^y+\mathcal K+\mu\mathcal I,\mathcal P^t+\mu'\mathcal I\rangle,\quad
\mathfrak s_{2.7}^\mu       =\langle\mathcal P^y+\mathcal K+\mu\mathcal P^t,\mathcal I\rangle,
\\
\hphantom{\rm 2D\colon}\quad
\smash{\mathfrak s_{2.8}^{\nu',\mu}}=\langle\mathcal P^y,\mathcal D+\nu'\mathcal P^t+\mu\mathcal I\rangle,\quad
\mathfrak s_{2.9}^\nu =\langle\mathcal P^y,\mathcal D+\nu\mathcal I\rangle,\quad
\\[1ex]
{\rm 3D}\colon\quad
\mathfrak s_{3.1}  =\langle\mathcal P^y,\mathcal P^t,\mathcal I\rangle,\quad
\mathfrak s_{3.2}  =\langle\mathcal D,\mathcal P^t,\mathcal I\rangle,\quad
\mathfrak s_{3.3}  =\langle\mathcal P^y+\mathcal K,\mathcal P^t,\mathcal I\rangle,
\\
\hphantom{\rm 3D\colon}\quad
\mathfrak s_{3.4}^{\nu,\mu}=\langle\mathcal P^y,\mathcal D+\nu\mathcal I,\mathcal P^t+\mu\mathcal I\rangle,\quad
\mathfrak s_{3.5}^\nu=\langle\mathcal P^y,\mathcal D+\nu\mathcal P^t,\mathcal I\rangle,\quad
\mathfrak s_{3.6}    =\langle\mathcal P^y,\mathcal D,\mathcal K\rangle,
\\[1ex]
{\rm 4D}\colon\quad
\mathfrak s_{4.1}    =\langle\mathcal P^y,\mathcal D,\mathcal P^t,\mathcal I\rangle,\quad
\mathfrak s_{4.2}^\mu=\langle\mathcal P^y,\mathcal D,\mathcal K,\mathcal P^t+\mu\mathcal I\rangle,\quad
\mathfrak s_{4.3}    =\langle\mathcal P^y,\mathcal D,\mathcal K,\mathcal I\rangle.
\end{gather*}
\end{lemma}

\begin{proof}
Recall that the algebra $\mathfrak g^{\rm ess}$ can be decomposed
into the direct sum of its Levi factor $\mathfrak f$ and radical $\mathfrak z$,
$\mathfrak g^{\rm ess}=\mathfrak f\oplus\mathfrak z$,
where $\mathfrak f\simeq{\rm sl}(2,\mathbb R)$ and $\mathfrak z\simeq 2A_1$, see Section~\ref{sec:FineFPMIA}.
Moreover, $G^{\rm ess}_{\rm id}=F_{\rm id}\times Z_{\rm id}$ and $Z_{\rm id}$ is the center of~$G^{\rm ess}_{\rm id}$,
and thus the actions of~$G^{\rm ess}_{\rm id}$ on~$\mathfrak z$ and of~$Z_{\rm id}$ on~$\mathfrak g^{\rm ess}$ are trivial.
To classify subalgebras of~$\mathfrak g^{\rm ess}$ up to the $G^{\rm ess}_{\rm id}$-equivalence,
we use the Goursat method, see \cite{pate1975a} for the first presentation of this method and~\cite{wint2004a} for further details.

Since the algebras $\mathfrak f$ and $\mathfrak z$ are the realizations of the algebras ${\rm sl}(2,\mathbb R)$ and $2A_1$, respectively,
optimal lists of their proper subalgebras are as follows:
\begin{gather*}
\mathfrak f\colon\quad
\langle\mathcal P^y\rangle,\quad
\langle\mathcal D\rangle,\quad
\langle\mathcal P^y+\mathcal K\rangle,\quad
\langle\mathcal P^y,\mathcal K\rangle;
\\
\mathfrak z\colon\quad
\langle\mathcal P^t+\mu\mathcal I\rangle_{\mu\in\mathbb R},\quad
\langle\mathcal I\rangle.
\end{gather*}
Let $\pi_{\mathfrak f}$ and $\pi_{\mathfrak z}$ denote the natural projections with respect to the vector space decomposition
$\mathfrak g^{\rm ess}=\mathfrak f\dotplus\mathfrak z$.
For any subalgebra $\mathfrak s\subset\mathfrak g$, we can set
$\pi_{\mathfrak f}\mathfrak s\in\{\langle0\rangle,\langle\mathcal P^y\rangle,
\langle\mathcal D\rangle,
\langle\mathcal P^y+\mathcal K\rangle,
\langle\mathcal P^y,\mathcal K\rangle,
\mathfrak f\}$,
modulo the $F_{\rm id}$-equivalences.
In addition, we have $\pi_{\mathfrak z}\mathfrak s,\mathfrak s\cap\mathfrak z\in\{\{0\},
\langle\mathcal P^t+\mu\mathcal I\rangle_{\mu\in\mathbb R},
\langle\mathcal I\rangle,
\mathfrak z\}$.
It is clear that  for any subalgebra $\mathfrak s\subset\mathfrak g^{\rm ess}$
the following values are invariants under the action of~$G^{\rm ess}_{\rm id}$:
$\dim\mathfrak s$,
$[\pi_{\mathfrak f}\mathfrak s]_{F_{\rm id}}^{}$,
$\pi_{\mathfrak z}\mathfrak s$
and $\mathfrak s\cap\mathfrak z$,
where $[\pi_{\mathfrak f}\mathfrak s]_{F_{\rm id}}^{}$ denotes the equivalence class of~$\pi_{\mathfrak f}\mathfrak s$
under the $F_{\rm id}$-equivalence.
This means that any two subalgebras with different corresponding tuples of the invariants introduced above cannot be equivalent.
Let us denote the quadruple
$\dim\mathfrak s$, $\dim\pi_{\mathfrak f}\mathfrak s$, $\dim\mathfrak s\cap\mathfrak z$ and $\dim\pi_{\mathfrak z}\mathfrak s$
by $n$, $\hat n$, $\check n$ and~$m$, respectively.
It is easy to see that $n=\hat n+\check n\leqslant\hat n+m$,
and
$\max(0,m-\hat n)\leqslant\check n\leqslant m\leqslant2$.

We use the above invariant objects for splitting the consideration into inequivalent cases.

\medskip\par\noindent	
$\boldsymbol{\hat n=0.}$
This means that $\mathfrak s\subseteq\mathfrak z$,
and thus the subalgebra~$\mathfrak s$ is equivalent to a subalgebra from the list of inequivalent subalgebras of $\mathfrak z$.
Therefore, we obtain the subalgebras $\mathfrak s_{1.1}^\mu$, $\mathfrak s_{1.2}$, and $\mathfrak s_{2.1}=\mathfrak z$.

\medskip

Before proceeding with classifying subalgebras $\mathfrak s$ with $\hat n=1$,
we note several simple observations.
Since $\hat n=1$, we have $n\leqslant3$.
\begin{itemize}\itemsep=0ex
\item[$(i)$]
If $n=1$, then $\check n=0$, i.e., $\mathfrak s\cap\mathfrak z=\{0\}$.
\item[$(ii)$]
If $n=2$, then $\check n=1$, and hence $\mathfrak s\cap\mathfrak z$ is a proper subalgebra of $\mathfrak s$.
Moreover, if in addition $m=1$, then $\mathfrak s\cap\mathfrak z=\pi_{\mathfrak z}\mathfrak s$,
which in its turn implies $\mathfrak s=\pi_{\mathfrak f}\mathfrak s\oplus\pi_{\mathfrak z}\mathfrak s$.
\item[$(iii)$]
If $n=3$, then $\check n=m=2$.
Hence $\mathfrak s\cap\mathfrak z=\pi_{\mathfrak z}\mathfrak s=\mathfrak z$,
which implies $\mathfrak s=\pi_{\mathfrak f}\mathfrak s\oplus\mathfrak z$.

\end{itemize}

\medskip\par\noindent	
$\boldsymbol{\pi_{\mathfrak f}\mathfrak s=\langle\mathcal P^y\rangle}$.
If $m=0$, then~$\mathfrak s=\mathfrak s_{1.4}^{0}$.
Further we assume $m>0$ and consider the cases with different values of~$n$ separately.

$n=1$.
For $\pi_{\mathfrak z}\mathfrak s=\langle\mathcal P^t+\mu\mathcal I\rangle$,
due to $\mathfrak s\cap\mathfrak z=\{0\}$ we obtain
$\mathfrak s=\langle\mathcal P^y+\alpha(\mathcal P^t+\mu\mathcal I)\rangle$ for some nonzero constant~$\alpha$.
Acting on $\mathfrak s$ by $\mathscr D(\tfrac12\ln|\alpha|)_*$,
we can always set $\alpha$ to be $\pm1$.
For $\pi_{\mathfrak z}\mathfrak s=\langle\mathcal I\rangle$, the arguments are analogous.

$n=2$.
For $\pi_{\mathfrak z}\mathfrak s=\langle\mathcal P^t+\mu\mathcal I\rangle$
or $\pi_{\mathfrak z}\mathfrak s=\langle\mathcal I\rangle$,
claim $(ii)$ implies
$\mathfrak s=\pi_{\mathfrak f}\mathfrak s\oplus\pi_{\mathfrak z}\mathfrak s$.
For $\pi_{\mathfrak z}\mathfrak s=\langle\mathcal P^t,\mathcal I\rangle$,
the cases with $\mathfrak s\cap\mathfrak z=\langle\mathcal P^t+\mu\mathcal I\rangle$
and $\mathfrak s\cap\mathfrak z=\langle\mathcal I\rangle$
correspond to the subalgebras $\mathfrak s=\langle\mathcal P^y+\nu\mathcal I,\mathcal P^t+\mu\mathcal I\rangle$
and $\mathfrak s=\langle\mathcal P^y+\nu\mathcal P^t,\mathcal I\rangle$,
where $\nu\ne0$.
Acting then on $\mathfrak s$ by $\mathscr D(\tfrac12\ln|\nu|)_*$, we can set $\nu=\pm1$.

$n=3$. We obtain $\mathfrak s=\pi_{\mathfrak f}\mathfrak s\oplus\mathfrak z$ by claim $(iii)$.

In total, this results in the families of subalgebras
$\mathfrak s_{1.3}^{\varepsilon,\mu}$, $\mathfrak s_{1.4}^{\delta}$,
$\mathfrak s_{2.2}^{\delta,\mu}$, $\mathfrak s_{2.3}^{\delta}$, $\mathfrak s_{3.1}$.

\medskip\par\noindent	
$\boldsymbol{\pi_{\mathfrak f}\mathfrak s=\langle\mathcal D\rangle}$.
If $m=0$, then $\mathfrak s=\mathfrak s^0_{1.6}$.
Further we assume $m>0$ and separately consider cases with different values of~$n$.

$n=1$.
The condition $\pi_{\mathfrak z}\mathfrak s=\langle\mathcal P^t+\mu\mathcal I\rangle$
implies that $\check n=0$, and thus $\mathfrak s=\langle\mathcal D+\nu\mathcal P^t+\mu\mathcal I\rangle$,
where $\nu\ne0$.
Acting on $\mathfrak s$ by $\mathscr Q^+(\pi/2)_*$, we can set $\nu>0$.
The consideration of the condition $\pi_{\mathfrak z}\mathfrak s=\langle\mathcal I\rangle$ is the same.
These two cases result in the family of subalgebras~\smash{$\mathfrak s_{1.5}^{\nu'\!,\mu}$} and~\smash{$\mathfrak s_{1.6}^{\nu'}$}
with $\nu'>0$ and $\mu\in\mathbb R$.

$n=2$.
If in addition $m=1$, then $\mathfrak s=\pi_{\mathfrak f}\mathfrak s\oplus\pi_{\mathfrak z}\mathfrak s$ due to claim $(ii)$,
and depending on $\pi_{\mathfrak z}\mathfrak s=\mathfrak s\cap\mathfrak z$,
we get the subalgebras~\smash{$\mathfrak s_{2.4}^{0,\mu}$} and~\smash{$\mathfrak s_{2.5}^0$}.
If $\pi_{\mathfrak z}\mathfrak s=\langle\mathcal P^t,\mathcal I\rangle$,
then the cases with
$\mathfrak s\cap\mathfrak z=\langle\mathcal P^t+\mu\mathcal I\rangle$ and
$\mathfrak s\cap\mathfrak z=\langle\mathcal I\rangle$
correspond to the families of subalgebras $\mathfrak s_{2.4}^{\nu,\mu}$
and $\mathfrak s_{2.5}^{\nu}$ with $\nu>0$ and $\mu\in\mathbb R$, respectively.

$n=3$. We straightforwardly obtain $\mathfrak s_{3.2}$.

\medskip\par\noindent	
$\boldsymbol{\pi_{\mathfrak f}\mathfrak s=\langle\mathcal P^y+\mathcal K\rangle}$.
This case is analogous to the previous cases
$\pi_{\mathfrak f}\mathfrak s=\langle\mathcal P^y\rangle$ and
$\pi_{\mathfrak f}\mathfrak s=\langle\mathcal D\rangle$,
except that scaling parameters or gauging their signs are impossible.
This gives subalgebras~\smash{$\mathfrak s_{1.7}^{\mu,\mu'}$}, \smash{$\mathfrak s_{2.6}^{\mu,\mu'}$},
$\mathfrak s_{2.7}^\mu$, $\mathfrak s_{3.3}$ with $\mu,\mu'\in\mathbb R$.

\medskip\par\noindent	
$\boldsymbol{\pi_{\mathfrak f}\mathfrak s=\langle\mathcal P^y,\mathcal D\rangle}.$
It is then clear that $2\leqslant n\leqslant4$.
If $m=0$, then $\mathfrak s=\mathfrak s^0_{2.9}$.
Further $m>0$.
We again consider cases with different values of~$n$ separately.

$n=2$.
Then $\check n=0$.
If $\pi_{\mathfrak z}\mathfrak s=\langle\mathcal P^t+\mu\mathcal I\rangle$,
then $\mathfrak s=\langle Q^1,Q^2\rangle$,
where
$Q^1=\mathcal P^y+\alpha(\mathcal P^t+\mu\mathcal I)$ and
$Q^2=\mathcal D+\beta(\mathcal P^t+\mu\mathcal I)$.
Since $[Q^1,Q^2]=\mathcal P^y\in\langle Q^1,Q^2\rangle$,
we derive $\alpha=0$,
which leads to the family of subalgebras $\mathfrak s_{2.8}^{\nu,\mu}$,
where we can additionally set $\nu>0$.
If $\pi_{\mathfrak z}\mathfrak s=\langle\mathcal I\rangle$, then
using analogous arguments, we get the family $\mathfrak s_{2.9}^\nu$ with $\nu>0$.
The case $\pi_{\mathfrak z}\mathfrak s=\langle\mathcal P^t,\mathcal I\rangle$
is impossible due to the commutation relations between involved vector fields.

$n=3$. Hence $\check n=1$.
If $m=1$, then $\mathfrak s\cap\mathfrak z=\pi_{\mathfrak z}\mathfrak s$,
i.e., the algebra~$\mathfrak s$ is the direct sum of its projections on $\mathfrak f$ and $\mathfrak z$,
which results in the subalgebras $\mathfrak s_{3.4}^{0,\mu}$ and $\mathfrak s_{3.5}^0$
for $\pi_{\mathfrak z}\mathfrak s=\langle\mathcal P^t+\mu\mathcal I\rangle$
and $\pi_{\mathfrak z}\mathfrak s=\langle\mathcal I\rangle$, respectively.

If $\pi_{\mathfrak z}\mathfrak s=\langle\mathcal P^t,\mathcal I\rangle$,
we still have two options for $\mathfrak s\cap\mathfrak z$,
either $\langle\mathcal P^t+\mu\mathcal I\rangle$ or $\langle\mathcal I\rangle$.
The corresponding subalgebras are
$\mathfrak s=\langle\mathcal P^y,\mathcal D+\nu\mathcal I,\mathcal P^t+\mu\mathcal I\rangle$
or $\mathfrak s=\langle\mathcal P^y,\mathcal D+\nu\mathcal P^t,\mathcal I\rangle$ with $\nu\ne0$,
where $\nu$ can be set greater than zero.
By this we obtain the subalgebras $\mathfrak s_{3.4}^{\nu,\mu}$ and $\mathfrak s_{3.5}^\nu$, where $\nu>0$.

$n=4$. It is obvious that then $\mathfrak s=\pi_{\mathfrak f}\mathfrak s\oplus\mathfrak z$.
Hence we obtain the subalgebra~$\mathfrak s_{4.1}$.

\medskip\par\noindent	
$\boldsymbol{\pi_{\mathfrak f}\mathfrak s=\langle\mathcal P^y,\mathcal D,\mathcal K\rangle}.$
If $n=3$, then the subalgebra~$\mathfrak s$ is a Levi factor of~$\mathfrak g^{\rm ess}$.
Therefore, by the Levi--Malcev theorem, it is equivalent with respect to an inner automorphism
to the subalgebra $\mathfrak s_{3.6}=\mathfrak f$.

In the case $n=4$ and $m=1$, we clearly have the factorization
$\mathfrak s=\pi_{\mathfrak f}\mathfrak s\oplus\pi_{\mathfrak z}\mathfrak s$
and obtain two families of inequivalent subalgebras, $\mathfrak s_{4.2}^\mu$ and $\mathfrak s_{4.3}$.
The case with $n=4$ and $m=2$ is impossible due to the commutation relations between involved vector fields.

\medskip\par
Each two subalgebra families among those listed in lemma's statement
correspond to different quadruples
$(\dim\mathfrak s,[\pi_{\mathfrak f}\mathfrak s]_{F_{\rm id}}^{},\pi_{\mathfrak z}\mathfrak s,\mathfrak s\cap\mathfrak z)$
of $G^{\rm ess}_{\rm id}$-invariants,
which means that there are no conjugations among subalgebras from different families.
The absence of conjugated subalgebras within each of the families follows from the above consideration.
\end{proof}

\begin{remark}
Due to the isomorphisms $\mathfrak g^{\rm ess}\simeq{\rm sl}(2,\mathbb R)\oplus2A_1\simeq{\rm}{\rm gl}(2,\mathbb R)\oplus A_1$,
Lemma~\ref{lem:IneqSubalgebrasUpToInn} also presents
the classifications of subalgebras of the algebras ${\rm sl}(2,\mathbb R)\oplus2A_1$ and ${\rm gl}(2,\mathbb R)\oplus A_1$
up to their inner automorphism groups.
The second decomposition also suggests another way of proving Lemma~\ref{lem:IneqSubalgebrasUpToInn}
that is based on the known list of inequivalent subalgebras of the algebra~${\rm gl}(2,\mathbb R)$
\cite{pate1977a,popo2003a}.
\end{remark}

\begin{corollary}\label{cor:GessIneqSubalgebras}
A complete list of $G^{\rm ess}$-inequivalent proper subalgebras of the algebra~$\mathfrak g^{\rm ess}$
is exhausted by the following subalgebras,
where $\delta'\in\{0,1\}$, $\mu\in\mathbb R$, $\nu\geqslant0$, $\nu'>0$:
\begin{gather*}
{\rm 1D}\colon\quad
\mathfrak s_{1.1}^\mu,\quad
\mathfrak s_{1.2},\quad
\mathfrak s_{1.3}^{1,\mu},\quad
\mathfrak s_{1.4}^{\delta'},\quad
\mathfrak s_{1.5}^{\nu',\mu},\quad
\mathfrak s_{1.6}^{\nu},\quad
\mathfrak s_{1.7}^{\nu',\mu},\quad
\mathfrak s_{1.7}^{0,\nu},
\\
{\rm 2D}\colon\quad
\mathfrak s_{2.1},\quad
\mathfrak s_{2.2}^{\delta',\mu},\quad
\mathfrak s_{2.3}^{\delta'},\quad
\mathfrak s_{2.4}^{\nu,\mu},\quad
\mathfrak s_{2.5}^\nu,\quad
\mathfrak s_{2.6}^{\nu,\mu},\quad
\mathfrak s_{2.7}^\nu,\quad
\mathfrak s_{2.8}^{\nu',\mu},\quad
\mathfrak s_{2.9}^{\nu},\quad
\\
{\rm 3D}\colon\quad
\mathfrak s_{3.1},\quad
\mathfrak s_{3.2},\quad
\mathfrak s_{3.3},\quad
\mathfrak s_{3.4}^{\nu,\mu},\quad
\mathfrak s_{3.5}^\nu,\quad
\mathfrak s_{3.6},
\\
{\rm 4D}\colon\quad
\mathfrak s_{4.1},\quad
\mathfrak s_{4.2}^\mu,\quad
\mathfrak s_{4.3}.
\end{gather*}
\end{corollary}

\begin{proof}
The group~$G^{\rm ess}$ is the semidirect product of its normal subgroup~$G^{\rm ess}_{\rm id}$
and its finite subgroup~$G^{\rm ess}_{\rm d}$ generated by the independent discrete symmetries~$\mathscr I'$ and~$\mathscr J'$
of the equation~\eqref{eq:FineFP}, which are given in Proposition~\ref{pro:FineFPDiscrSyms}.
If subalgebras are $G^{\rm ess}_{\rm id}$-equivalent, then they are clearly $G^{\rm ess}$-equivalent.
Hence a complete list of $G^{\rm ess}$-inequivalent proper subalgebras of the algebra $\mathfrak g^{\rm ess}$
is contained within the list presented in Lemma~\ref{lem:IneqSubalgebrasUpToInn}.
Moreover, subalgebras from different families in the latter list are $G^{\rm ess}$-inequivalent
since the quadruple of $G^{\rm ess}_{\rm id}$-invariants
$(\dim\mathfrak s,[\pi_{\mathfrak f}\mathfrak s]_{F_{\rm id}}^{},\pi_{\mathfrak z}\mathfrak s,\mathfrak z\cap\mathfrak s)$
is invariant under the action of $G^{\rm ess}$ as well.
Among the transformations~$\mathscr I'$ and~$\mathscr J'$,
the only pushforward of~$\mathscr J'$ acts nontrivially on $\mathfrak g^{\rm ess}$,
and this action is nonidentity only for the basis elements~$\mathcal P^y$ and~$\mathcal K$,
$\mathscr J'_*\mathcal P^y=-\mathcal P^y$ and $\mathscr J'_*\mathcal K  =-\mathcal K$.
Hence acting by~$G^{\rm ess}_{\rm d}$ on the subalgebras listed in Lemma~\ref{lem:IneqSubalgebrasUpToInn}
establishes the following $G^{\rm ess}$-equivalences:
\begin{gather*}
\mathfrak s_{1.3}^{-1,-\mu}\sim\mathfrak s_{1.3}^{1,\mu},\quad
\mathfrak s_{1.4}^{-1}\sim\mathfrak s_{1.4}^1,\quad
\mathfrak s_{1.7}^{-\mu,-\mu'}\sim\mathfrak s_{1.7}^{\mu,\mu'},
\\
\mathfrak s_{2.2}^{-1,\mu}\sim\mathfrak s_{2.2}^{1,\mu},\quad
\mathfrak s_{2.3}^{-1}\sim\mathfrak s_{2.3}^1,\quad
\mathfrak s_{2.6}^{-\mu,\mu'}\sim\mathfrak s_{2.6}^{\mu,\mu'},\quad
\mathfrak s_{2.7}^{-\mu}\sim\mathfrak s_{2.7}^{\mu}.
\end{gather*}
For each of the above pairs, we exclude one of its elements from the list in Lemma~\ref{lem:IneqSubalgebrasUpToInn},
which gives a complete list of $G^{\rm ess}$-inequivalent proper subalgebras of the algebra $\mathfrak g^{\rm ess}$.
\end{proof}

\section{Lie reductions}\label{sec:LieReductions}

Throughout this section, the subscripts~1 and~2 of the dependent variable $w$
denote the derivatives with respect to~$z_1$ and~$z_2$, respectively.
We use the complex label $d.i^*$ to mark the Lie reduction
corresponding to the subalgebra $\mathfrak s_{d.i}^*$ from Corollary~\ref{cor:GessIneqSubalgebras},
where $d$, $i$ and~$*$ are the dimension of this subalgebra,
its number in the list of $d$-dimensional subalgebras of the algebra~$\mathfrak g^{\rm ess}$,
and the values of subalgebra parameters for a subalgebra family, respectively.
We omit the superscript in the label whenever it is not essential.

\medskip\par\noindent
{\bf Codimension one.}\
All the inequivalent one-dimensional subalgebras of~$\mathfrak g^{\rm ess}$
that are listed in Corollary~\ref{cor:GessIneqSubalgebras}, except the subalgebra~$\mathfrak s_{1.2}$,
satisfy the transversality condition and thus are appropriate for using for Lie reductions of the equation~\eqref{eq:FineFP}.
(The space of $\mathfrak s_{1.2}$-invariant solutions of~\eqref{eq:FineFP} is exhausted by the trivial identically zero solution.)
For each of these subalgebras, we construct the corresponding Lie ansatz, i.e.,
the expression for~$u$ in terms of the new unknown function $w(z_1,z_2)$
of two invariant independent variables $z_1$ and~$z_2$,
and derive the associated reduced equation.
We merge the subalgebra families $\{\mathfrak s_{1.5}^{\nu',\mu}\}$ and $\{\mathfrak s_{1.6}^{\nu}\}$
to $\{\mathfrak s_{1.5}^{\nu,\mu}\}$, assuming that $\mu\geqslant0$ if $\nu=0$.
The results are presented in Table~\ref{tab:LieReductionsOfCodim1}.

\begin{table}[!ht]\footnotesize
\begin{center}
\caption{$G^{\rm ess}$-inequivalent codimension-one Lie reductions of the equation~\eqref{eq:FineFP}}
\label{tab:LieReductionsOfCodim1}${ }$\\[-1ex]
\renewcommand{\arraystretch}{1.6}
\begin{tabular}{|l|c|c|c|l|}
\hline
\hfil no.		& $u$                             & $z_1$            & $z_2$     & \hfil Reduced equation\\
\hline
1.1$^\mu$ 		& ${\rm e}^{\mu t}w$              & $y$              & $x$             & $z_2w_1=z_2^2w_{22}-\mu w$     \\
1.3$^{1,\mu}$   & ${\rm e}^{\mu t}w$              & $t-y$            & $x$             & $(1-z_2)w_1=z_2^2w_{22}-\mu w$ \\
1.4$^{\delta'}$ & ${\rm e}^{\delta' y}w$          & $t$              & $x$             & $w_1=z_2^2w_{22}-\delta'z_2w$  \\
1.5$^{\nu,\mu}$         & $|y|^\mu w$                     & $t-\nu\ln|y|$    & $x/y$           & $(1-\nu z_2)w_1=z_2^2w_{22}+z_2^2w_2-\mu z_2w$\\
1.7$^{\nu,\mu}$   & ${\rm e}^{\mu\arctan(y)-\frac{xy}{y^2+1}}w$& $t-\nu\arctan y$&$\dfrac x{y^2+1}$ &$(1-\nu z_2)w_1=z_2^2w_{22}+z_2(z_2-\mu)w$\\[1.2ex]
\hline
\end{tabular}
\end{center}
Here $\delta'\in\{0,1\}$, $\mu\in\mathbb R$, $\nu\geqslant0$, 
and $\mu\geqslant0$ in Cases~1.5 and~1.7 with $\nu=0$;
$w=w(z_1,z_2)$ is the new unknown function of the new independent variables $(z_1,z_2)$.
\end{table}

The reduced equations are linear homogeneous second-order evolution equations in two independent variables.
The group classification of these equations is well known
and was first obtained by Sophus Lie himself over~$\mathbb C$ in~\cite{lie1881a},
see also~\cite[Section 2]{popo2008a} for the modern view on the problem.
We exploit this classification for constructing point transformations
that map the reduced equations to the canonical representatives of their equivalence class,
which are (1+1)-dimensional linear heat equations with potentials.

For each of the reduced equations, $\mathcal R$,
its maximal Lie invariance algebra $\mathfrak g_{\mathcal R}$
contains the infinite-dimensional abelian ideal $\{h(z_1,z_2)\p_w\}$,
where $h$ runs through the solution set of~$\mathcal R$.
This ideal is merely associated with the linear superposition of solutions of~$\mathcal R$
and thus can be assumed to be a trivial part of the algebra $\mathfrak g_{\mathcal R}$.
Moreover, the algebra $\mathfrak g_{\mathcal R}$ splits over it,
and its complement in $\mathfrak g_{\mathcal R}$ is
a finite-dimensional subalgebra $\mathfrak a_{\mathcal R}$,
called the essential Lie invariance algebra of~$\mathcal R$.
The algebras $\mathfrak a_{\mathcal R}$ for specific reduced equations~$\mathcal R$ are the following:
\begin{gather*}
\mathfrak a_{1.3}^\mu
=\mathfrak a_{1.4}^1
=\mathfrak a_{1.5}^{\nu,\mu}
=\mathfrak a_{1.7}^{\nu,\mu}=\big\langle\p_{z_1},\,w\p_w\big\rangle,
\\[.5ex]
\mathfrak a_{1.1}^\mu
=\big\langle
\p_{z_1},\,
z_1\p_{z_1}+z_2\p_{z_2},\,
z_1^2\p_{z_1}+2z_1z_2\p_{z_2}-wz_2\p_w,\,
w\p_w
\big\rangle
\quad\text{if}\quad\mu\ne-\tfrac3{16},
\\[.5ex]
\mathfrak a_{1.1}^{\mu_0}
=\big\langle
\p_{z_1},\,
z_1\p_{z_1}+z_2\p_{z_2}-\tfrac12w\p_w,\,
z_1^2\p_{z_1}+2z_1z_2\p_{z_2}-z_2w\p_w,
\\
\hphantom{\mathfrak a_{1.1}^{\mu_0}=\langle}
|z_2|^{-\frac12}\big(4z_2\p_{z_2}+w\p_w\big),\,
|z_2|^{-\frac12}\big(4z_1z_2\p_{z_2}+(z_1-4z_2)w\p_w\big),\,
w\p_w
\big\rangle
\ \ \text{with}\ \ \mu_0:=-\tfrac3{16},
\\[.5ex]
\mathfrak a_{1.4}^0
=\big\langle
\p_{z_1},\,
4z_1\p_{z_1}+2z_2\ln|z_2|\p_{z_2}-(z_1-\ln|z_2|)w\p_w,
\\
\hphantom{\mathfrak a_{1.4}^0=\langle}
4z_1^2\p_{z_1}+4z_1z_2\ln|z_2|\p_{z_2}-\big((z_1-\ln|z_2|)^2+2z_1\big)w\p_w,
\\
\hphantom{\mathfrak a_{1.4}^0=\langle}
2z_1z_2\p_{z_2}+(z_1-\ln|z_2|)w\p_w,\,
z_2\p_{z_2},\,
w\p_w
\big\rangle.
\end{gather*}

Since $\dim\mathfrak a_{1.1}^{\mu_0}=\dim\mathfrak a_{1.4}^0>\dim\mathfrak g^{\rm ess}$,
the equation~\eqref{eq:FineFP} definitely admits hidden symmetries related to Lie reductions $1.1^{\mu_0}$ and $1.4^0$.
Recall that a hidden symmetry of a differential equation is a symmetry of its reduced equation
that is not induced by a symmetry of the original equation.
To find hidden symmetries, one usually describes the subalgebra of induced symmetries
in the Lie invariance algebra of the reduced equation.
The algebra of induced essential Lie symmetries of reduced equation $1.i$ is isomorphic to
the quotient algebra ${\rm N}_{\mathfrak g^{\rm ess}}(\mathfrak s_{1.i})/\mathfrak s_{1.i}$.
Hence to check which Lie symmetries of reduced equation~$1.i$ are induced by Lie symmetries of~\eqref{eq:FineFP},
we compute the normalizer ${\rm N}_{\mathfrak g^{\rm ess}}(\mathfrak s_{1.i})$
of the subalgebra~$\mathfrak s_{1.i}$ in the algebra $\mathfrak g^{\rm ess}$,
\begin{gather*}
{\rm N}_{\mathfrak g^{\rm ess}}(\mathfrak s_{1.1}^\mu)=\mathfrak g^{\rm ess},\ \;
{\rm N}_{\mathfrak g^{\rm ess}}(\mathfrak s_{1.4}^0)=\langle\mathcal P^y,\mathcal D,\mathcal P^t,\mathcal I\rangle,\ \;
{\rm N}_{\mathfrak g^{\rm ess}}(\mathfrak s_{1.3}^{1,\mu})
={\rm N}_{\mathfrak g^{\rm ess}}(\mathfrak s_{1.4}^1)
=\langle\mathcal P^y,\mathcal P^t,\mathcal I\rangle,\ \;
\\
{\rm N}_{\mathfrak g^{\rm ess}}(\mathfrak s_{1.5}^{\nu,\mu})
=\langle\mathcal D,\mathcal P^t,\mathcal I\rangle,\ \;
{\rm N}_{\mathfrak g^{\rm ess}}(\mathfrak s_{1.7}^{\nu,\mu})
=\langle\mathcal P^y+\mathcal K,\mathcal P^t,\mathcal I\rangle.
\end{gather*}

The subalgebras of essential induced symmetries
in $\mathfrak a_{1.1}^\mu$, $\mathfrak a_{1.2}$, \dots, $\mathfrak a_{1.7}^{\nu,\mu}$ are
\begin{gather*}
\tilde{\mathfrak a}_{1.1}^\mu=\langle\p_{z_1},z_1\p_{z_1}+z_2\p_{z_2},2z_1z_2\p_{z_2}+z_1^2\p_{z_1}-z_2w\p_w,w\p_w\rangle,\quad
\tilde{\mathfrak a}_{1.4}^0  =\langle\p_{z_1},z_2\p_{z_2},w\p_w\rangle,
\\
\tilde{\mathfrak a}_{1.3}^{1,\mu}
=\tilde{\mathfrak a}_{1.4}^1
=\tilde{\mathfrak a}_{1.5}^{\nu,\mu}
=\tilde{\mathfrak a}_{1.7}^{\nu,\mu}
=\langle\p_{z_1},w\p_w\rangle.
\end{gather*}
Hence the genuine hidden symmetries of~\eqref{eq:FineFP}
are contained only in the set complements $\mathfrak a_{1.1}^{\mu_0}\setminus\tilde{\mathfrak a}_{1.1}^{\mu_0}$
and $\mathfrak a_{1.4}^0\setminus\tilde{\mathfrak a}_{1.4}^0$.
This means that the further Lie reductions of the reduced equations, except~$1.1^{\mu_0}$ and~$1.4^0$
lead only to solutions of~\eqref{eq:FineFP}
that can be obtained by codimensions-two Lie reductions of~\eqref{eq:FineFP}.

The dimension of the essential Lie invariance algebras of each of the reduced equations is greater than one.
Therefore, we can map these equations by point transformations
to (1+1)-dimensional linear heat equations with potentials,
which are functions of the new variable~$\tilde z_2$,
i.e., to the equations of the form $\tilde w_{\tilde z_1}=\tilde w_{\tilde z_2\tilde z_2}+V(\tilde z_2)\tilde w$.
Here and in what follows the subscripts~1 and~2 of~$\tilde w$ denote the derivatives of~$\tilde w$
with respect to~$\tilde z_1$ and~$\tilde z_2$, respectively.
The transformations and the mapped equations are
\begin{description}

\item[$1.1^\mu$:\ ]
$\tilde z_1=\varepsilon z_1,\ \
\tilde z_2=2\sqrt{|z_2|},\ \
\tilde w(\tilde z_1,\tilde z_2)=|z_2|^{-\frac14}w,\quad
\tilde w_1=\tilde w_{22}-\dfrac{4\mu+3/4}{\tilde z_2^2}\tilde w$;

\item[$1.3^\mu$, $z_2<1$:\ ]
$\tilde z_1=z_1,\ \
\tilde z_2=2\sqrt{1-z_2}+\ln\left|\dfrac{\sqrt{1-z_2}-1}{\sqrt{1-z_2}+1}\right|,
\ \
\tilde w=\dfrac{(1-z_2)^{1/4}}{|z_2|^{1/2}}w$,
\\[1ex]
$\tilde w_1=\tilde w_{22}-\dfrac1{16}\left(\dfrac{16\mu+3}{z_2-1}-\dfrac6{(z_2-1)^2}-\dfrac5{(z_2-1)^3}\right)\tilde w$;

\item[$1.3^\mu$, $z_2>1$:\ ]
$\tilde z_1=-z_1,\ \
\tilde z_2=2\sqrt{z_2-1}-2\arctan\sqrt{z_2-1},
\ \
\tilde w=\dfrac{(z_2-1)^{1/4}}{z_2^{1/2}}w$,
\\
$\tilde w_1=\tilde w_{22}+\dfrac1{16}\left(\dfrac{16\mu+3}{z_2-1}-\dfrac6{(z_2-1)^2}-\dfrac5{(z_2-1)^3}\right)\tilde w$;

\item[$1.4^\delta$:]
$\tilde z_1=z_1,\ \
\tilde z_2=\ln|z_2|,\ \
\tilde w={\rm e}^{z_1/4}|z_2|^{-1/2}w,\ \
\tilde w_1=\tilde w_{22}-\delta\tilde\varepsilon{\rm e}^{\tilde z_2}\tilde w$\quad with\quad $\tilde\varepsilon=\sgn z_2$;

\item[$1.5^{\nu',\mu}$, $\nu'z_2<1$:]
$\tilde z_1=z_1,\ \
\tilde z_2=2\sqrt{1-\nu'z_2}+\ln\left|\dfrac{\sqrt{1-\nu'z_2}-1}{\sqrt{1-\nu'z_2}+1}\right|,\ \
\tilde w=\dfrac{{\rm e}^{z_2/2}}{|z_2|^{1/2}}(1-\nu'z_2)^{1/4}w$,
\\[1ex]
$\tilde w_1=\tilde w_{22}+\dfrac1{16}\left(\dfrac{4z_2^2+16\mu z_2+3}{\nu'z_2-1}
-\dfrac6{(\nu'z_2-1)^2}-\dfrac5{(\nu'z_2-1)^3}\right)\tilde w$;

\item[$1.5^{\nu',\mu}$, $\nu'z_2>1$:]
$\tilde z_1=-z_1,\ \
\tilde z_2=2\sqrt{\nu'z_2-1}-2\arctan\sqrt{\nu'z_2-1},\ \
\tilde w=\dfrac{{\rm e}^{z_2/2}}{z_2^{1/2}}(\nu'z_2-1)^{1/4}w$,
\\
$\tilde w_1=\tilde w_{22}-\dfrac1{16}\left(\dfrac{4z_2^2+16\mu z_2+3}{\nu'z_2-1}
-\dfrac6{(\nu'z_2-1)^2}-\dfrac5{(\nu'z_2-1)^3}\right)\tilde w$;

\item[$1.5^{0,\nu}$:]
$\tilde z_1=z_1,\ \
\tilde z_2=\ln|z_2|,\ \
\tilde w={\rm e}^{z_1/4+z_2/2}|z_2|^{-1/2}w,\\[1ex]
\tilde w_1=\tilde w_{22}-(\tilde\varepsilon\nu{\rm e}^{\tilde z_2}+\frac14{\rm e}^{2\tilde z_2})\tilde w
$ \ with \ $\tilde\varepsilon=\sgn z_2$;

\item[$1.7^{\nu',\mu}$, $\nu'z_2<1$:]
$\tilde z_1=z_1,\ \
\tilde z_2 =2\sqrt{1-\nu'z_2}+\ln\left|\dfrac{\sqrt{1-\nu'z_2}-1}{\sqrt{1-\nu'z_2}+1}\right|,\ \
\tilde w   =\dfrac{(1-\nu'z_2)^{1/4}}{|z_2|^{1/2}}w$,
\\[1ex]
$\tilde w_1=\tilde w_{22}
+\dfrac1{16}\left(\dfrac{16z_1(\mu-z_2)+3}{\nu z_2-1}-\dfrac6{(\nu z_2-1)^2}-\dfrac5{(\nu z_2-1)^3}\right)\tilde w$;

\item[$1.7^{\nu',\mu}$, $\nu'z_2>1$:]
$\tilde z_1=z_1,\ \
\tilde z_2 =2\sqrt{\nu'z_2-1}-2\arctan\sqrt{\nu'z_2-1},\ \
\tilde w   =\dfrac{(\nu'z_2-1)^{1/4}}{z_2^{1/2}}w$,
\\[1ex]
$\tilde w_1=\tilde w_{22}
-\dfrac1{16}\left(\dfrac{16z_1(\mu-z_2)+3}{\nu z_2-1}-\dfrac6{(\nu z_2-1)^2}-\dfrac5{(\nu z_2-1)^3}\right)\tilde w$;

\item[$1.7^{0,\nu}$:]
$\tilde z_1=z_1,\ \
\tilde z_2=\ln|z_2|,\ \
\tilde w=\dfrac{{\rm e}^{z_1/4}}{|z_2|^{1/2}}w,\ \
\tilde w_1=\tilde w_{22}+({\rm e}^{2\tilde z_2}-\tilde\varepsilon\nu{\rm e}^{\tilde z_2})\tilde w
$
\ with \ $\tilde\varepsilon:=\sgn z_2$.
\end{description}

The potentials in Cases~$1.3^\mu$, $1.5^{\nu',\mu}$ and~$1.7^{\nu',\mu}$
are defined as expressions with~$z_2=\tilde Z^2(\tilde z_2)$,
where $\tilde Z^2$ is the inverse of the function in the $z_2$-component $\tilde z_2=Z^2(z_2)$
of the corresponding transformation.

In several cases, where the dimension of the essential Lie invariance algebras of
the associated reduced equations is even greater than two,
the corresponding mapped equations are of especially nice form.

Thus, the essential Lie invariance algebras $\mathfrak a_{1.1}^{\mu_0}$ and $\mathfrak a_{1.4}^0$
are isomorphic to the special Galilei algebra,
which agrees with the fact that the mapped equations in Cases~$1.1^{\mu_0}$ and $1.4^0$
coincide with the (1+1)-dimensional linear heat equation.
This leads to the solutions of the equation~\eqref{eq:FineFP}
\begin{gather}\label{eq:SolFineFPHeat1}
\solution u={\rm e}^{-\frac3{16}t}|x|^{1/4}\theta(\varepsilon y,2\sqrt{|x|}),
\\ \label{eq:SolFineFPHeat2}
\solution u={\rm e}^{-\frac14t}|x|^{1/2}\theta(t,\ln|x|),
\end{gather}
where $\theta$ is an arbitrary solution of the equation $\theta_1=\theta_{22}$.
An enhanced complete collection of inequivalent Lie invariant solutions of this equation
was presented in~\cite[Section A]{vane2021a}, following Examples 3.3 and 3.17 in~\cite{olve1993A}.

The essential Lie invariance algebra of the equation~$1.1^\mu$ with $\mu\ne-\frac3{16}$ is four-dimensional.
Hence the corresponding mapped equation is the (1+1)-dimensional linear heat equation with inverse square potential.
This gives wide families of solutions of the equation~\eqref{eq:FineFP} that are of the form
\begin{gather}\label{eq:SolFineFPHeatWithZ2-2Pot}
\solution u={\rm e}^{\mu t}|x|^{1/4}\vartheta^{4\mu+\frac34}(\varepsilon y,2\sqrt{|x|}),
\end{gather}
where $\vartheta^{\tilde\mu}(\tilde z_1,\tilde z_2)$ is an arbitrary solution of the equation
$\tilde w_1=\tilde w_{22}-\tilde\mu\tilde z_2^{-2}\tilde w$ with $\tilde\mu\ne0$.
Wide families of solutions for such equations were found in~\cite[Section~A]{kova2023a}.

\medskip\par\noindent
{\bf Codimensions two and three.}\
Among the inequivalent two-dimensional subalgebras listed in Corollary~\ref{cor:GessIneqSubalgebras}
the subalgebras $\mathfrak s_{2.1}$, $\mathfrak s_{2.3}^{\delta'}$,
$\mathfrak s_{2.5}^\nu$, and $\mathfrak s_{2.7}^\nu$ are inappropriate for constructing a Lie ansatz.
Each of the other listed two-dimensional subalgebras contains
the subalgebra $\mathfrak s_{1.1}^\mu$ or the subalgebra $\mathfrak s_{1.4}^0$.
Therefore, each of the solutions constructed using the corresponding Lie codimension-two reductions
is $G^{\rm ess}$-equivalent to a solution
from one of the families~\eqref{eq:SolFineFPHeat1}--\eqref{eq:SolFineFPHeatWithZ2-2Pot}.

An analogous claim is true for reductions to algebraic equations.
The only listed three-dimensional subalgebras
that are suitable for Lie reductions are $\mathfrak s_{3.4}^{\nu,\mu}$ and $\mathfrak s_{3.6}$.
However, any of them contains the subalgebra $\mathfrak s_{1.4}^0$.
Hence the corresponding reduction gives no essentially new solution in addition to the already constructed ones.

\section{Generating new solutions from known ones}\label{sec:SolutionGeneration}

Point symmetries of the fine Kolmogorov backward equation~\eqref{eq:FineFP} can also be exploited
in the course of generating new solutions of this equation from the known ones.
Recall that by definition a point symmetry of the system of differential equations
is a point transformation that preserves its solution set.
Thus, in the notation of Theorem~\ref{thm:FineFPSymGroup},
pulling back a given solution $u=h(t,x,y)$ of~\eqref{eq:FineFP}
by an arbitrary element of the group~$G$ gives
\begin{gather*}
\tilde u=\sigma\exp\left(-\frac{\gamma}{\alpha-\gamma y}x\right)
h\left(t-\lambda,
\frac{\alpha\delta-\beta\gamma}{(\alpha-\gamma y)^2}x,
\frac{\delta y-\beta}{\alpha-\gamma y}\right)
+f(t,x,y),
\end{gather*}
which is the most general form of the $G$-equivalent counterparts of the seed solution $u=h(t,x,y)$.

Due to the linearity of the equation~\eqref{eq:FineFP},
there is one more way of generating its solutions from the known ones within the framework of symmetry analysis,
which uses linear generalized symmetries of~\eqref{eq:FineFP}.
To the best of our knowledge, the study of generalized symmetries of the equation~\eqref{eq:FineFP}
has not been initiated in the literature.
The thorough investigation of the generalized symmetry algebra of this equation
involves cumbersome and sophisticated computations and can become a topic of a separate research.
Nevertheless, there are a number of arguments on the description of this algebra to be expected.
See~\cite[Chapter~5]{olve1993A}, especially pp.~305--306, and~\cite{kova2024b}
for required definitions, notation and related results.
In particular, we refer to the differential operators in total derivatives
that are associated with the linear Lie-symmetry vector fields of the equation~\eqref{eq:FineFP}
as to the {\it Lie-symmetry operators} of this equation.
Thus, the Lie-symmetry vector fields $-\mathcal P^y$, $-\mathcal D$, $-\mathcal K$, $-\mathcal P^t$ and~$\mathcal I$
correspond to the operators
\begin{gather*}
\mathrm P^y:=\mathrm D_y,\quad
\mathrm D  :=x\mathrm D_x+y\mathrm D_y,\quad
\mathrm K  :=2xy\mathrm D_x+y^2\mathrm D_y+x,\quad
\mathrm P^t:=\mathrm D_t
\end{gather*}
and the identity operator.
The equation~\eqref{eq:FineFP} is also extraordinary by the fact that
its operator~$L:=\mathrm D_t+x\mathrm D_y-x^2\mathrm D_x^2$
can be represented in terms of its Lie-symmetry operators,
\[
L=\mathrm P^t-\mathrm D^2+\frac12(\mathrm P^y\mathrm K+\mathrm K\mathrm P^y).
\]
The operator~$\mathrm P^t$ is equivalent on the solutions of~\eqref{eq:FineFP} to the element
\[
\hat{\mathrm P}^t
:=\mathrm D^2-\frac12(\mathrm P^y\mathrm K+\mathrm K\mathrm P^y)
=\mathrm D^2-\mathrm D-\mathrm K\mathrm P^y
\]
of the associative algebra $\Upsilon_{\mathfrak f}$ generated by
$\mathrm P^y$, $\mathrm D$ and~$\mathrm K$ since $\mathrm P^t-\hat{\mathrm P}^t=L$.
Moreover, the operator~$\hat{\mathrm P}^t$ is a Casimir operator of the Levi factor
$\mathfrak f=\langle\mathcal P^y,\mathcal D,\mathcal K\rangle$ of the algebra~$\mathfrak g^{\rm ess}$,
see Section~\ref{sec:FineFPMIA}.

Based on a preliminary analysis in~\cite{kova2024b}, we conjecture the following.

\begin{conjecture}\label{conj:GenSyms}
Up to the equivalence of generalized symmetries
and neglecting the Lie symmetries associated with the linear superposition of solutions,
generalized symmetries of the fine Kolmogorov backward equation~\eqref{eq:FineFP}
are exhausted by its reduced linear generalized symmetries.
The Lie algebra~$\hat\Lambda$ of these symmetries is generated from the simplest linear generalized symmetry $u\p_u$
by acting with the recursion operators~$\mathrm P^y$, $\mathrm D$ and~$\mathrm K$ of the equation~\eqref{eq:FineFP},
$\hat\Lambda=\{(\mathrm Qu)\p_u\mid\mathrm Q\in\Upsilon_{\mathfrak f}\}$.
\end{conjecture}

In fact, the operators $\mathrm P^y$ and~$\mathrm K$ suffice for generating the algebra~$\Upsilon_{\mathfrak f}$
since $\mathrm D=\frac12(\mathrm P^y\mathrm K-\mathrm K\mathrm P^y)$.
However, considering~$\{\mathrm P^y,\mathrm D,\mathrm K\}$ as the canonical generating set for~$\Upsilon_{\mathfrak f}$
is natural and beneficial since fixing then an ordering in this set,
we get the associated canonical basis of~$\Upsilon_{\mathfrak f}$.

\begin{lemma}\label{lem:BasisUpsilon}
The monomials $\mathbf Q^\alpha:=(\mathrm Q^1)^{\alpha_1}(\mathrm Q^2)^{\alpha_2}(\mathrm Q^3)^{\alpha_3}$,
where $\alpha=(\alpha_1,\alpha_2,\alpha_3)\in\mathbb N_0^{\,\,3}$ and
$(\mathrm Q^1,\mathrm Q^2,\mathrm Q^3)$ is any fixed ordering of~$\mathrm P^y$, $\mathrm D$ and~$\mathrm K$,
constitute a basis of the algebra $\Upsilon_{\mathfrak f}$.
\end{lemma}

\begin{proof}
The algebra	$\Upsilon_{\mathfrak f}$ has a canonical filtration
associated with the order $\ord\mathrm Q$ of the differential operators $\mathrm Q$,
\[
\Upsilon_{\mathfrak f}=\bigcup_{k=0}^\infty\Upsilon_{\mathfrak f,k}\quad
\text{with}\quad
\Upsilon_{\mathfrak f,k}=\{\mathrm Q\in\Upsilon_{\mathfrak f}\mid\ord\mathrm Q\leqslant k\},\quad
\Upsilon_{\mathfrak f,k}\Upsilon_{\mathfrak f,l}\subset\Upsilon_{\mathfrak f,k+l},\quad k,l\in\mathbb N_0.
\]
It suffices to show that the monomials
$\mathbf Q^\alpha$ with $|\alpha|:=\alpha_1+\alpha_2+\alpha_3\leqslant k$ constitute a basis of~$\Upsilon_{\mathfrak f,k}$.
It is clear that these monomials span $\Upsilon_{\mathfrak f,k}$ for any $k\in\mathbb N_0$,
since any monomial $\mathrm S^1\cdots\mathrm S^l$ with $\mathrm S^1,\dots,\mathrm S^k\in\{\mathrm P^y,\mathrm D,\mathrm K\}$
can be represented as a linear combination of the monomials
$\mathbf Q^\alpha$ with $|\alpha|\leqslant l$
due to the commutation relations between operators in $\{\mathrm P^y,\mathrm D,\mathrm K\}$.

It remains to establish the linear independence of the monomials $\mathbf Q^\alpha$ with $|\alpha|\leqslant k$,
which we will do by the induction with respect to the order~$k$.
The induction basis with $k=0$ is obvious.
Suppose that $\mathbf Q^\alpha$, $|\alpha|\leqslant k$ are linearly independent and
\begin{gather}\label{eq:LinIndep}
\mathrm S:=\sum_{|\alpha|\leqslant k+1}\lambda_\alpha\mathbf Q^\alpha=0
\end{gather}
for some $\lambda_\alpha\in\mathbb R$, $|\alpha|\leqslant k+1$.
Then ${\rm e}^{-\mu x-\nu y}\mathrm S{\rm e}^{\mu x+\nu y}$
is an identically vanishing polynomial in $(x,y,\mu,\nu)$.
The part of this polynomial that includes all the terms of order~$k+1$ with respect to $(\mu,\nu)$,
i.e., $\sum_{|\alpha|=k+1}\lambda_\alpha z_1^{\alpha_1}z_2^{\alpha_2}z_3^{\alpha_3}$,
where $z_1:=\nu$, $z_2:=\mu x+\nu y$ and $z_3:=\nu y^2+2\mu xy+x$, also vanishes.
Since $z_1$, $z_2$ and $z_3$ are functionally independent as functions of $(x,y,\mu,\nu)$,
this implies that $\lambda_\alpha=0$ if $|\alpha|=k+1$.
Invoking the induction hypothesis, the equality~\eqref{eq:LinIndep}
implies $\lambda_\alpha=0$ for every $\alpha$ with $|\alpha|\leqslant k+1$.
\end{proof}

\begin{remark}\label{rem:BasisUpsilon}
In Lemma~\ref{lem:BasisUpsilon},
the operators~$\mathrm P^y$, $\mathrm D$ and~$\mathrm K$ can be replaced
by linear combinations of the form $a_{i1}\mathrm P^y+a_{i2}\mathrm D+a_{i3}\mathrm K+a_{i0}$,
where $\det(a_{ij})_{i,j=1,2,3}\ne0$.
A more sophisticated basis is convenient for a part of the further consideration.
We split the space~$\Upsilon_{\mathfrak f}$ into two subspaces~$\Gamma_1$ and~$\Gamma_2$
that are spanned by the monomials $\mathrm K^{\alpha_1}(\mathrm P^y)^{\alpha_2}\mathrm D^{\alpha_3}$
with $\alpha_1\leqslant\alpha_2$ and $\alpha_1>\alpha_2$, respectively.
In view of the relations
$\mathrm K\mathrm P^y=\mathrm D^2-\mathrm D-\hat{\mathrm P}^t$ and
$\mathrm P^y\mathrm D=(\mathrm D+1)\mathrm P^y$,
the subspace~$\Gamma_1$ coincides with the space of polynomials in $(\hat{\mathrm P}^t,\mathrm P^y,\mathrm D)$
and the subspace~$\Gamma_2$ possesses the bases constituted by the monomials
$\mathrm K^{\alpha_1}(\hat{\mathrm P}^t)^{\alpha_2}\mathrm D^{\alpha_3}$ with $\alpha_1>0$.
\end{remark}

\begin{corollary}\label{cor:IsomToUnivEnv}
The algebra $\Upsilon_{\mathfrak f}$ is isomorphic to the universal enveloping algebra $\mathfrak U(\mathrm{sl}(2,\mathbb R))$
of the real rank-two special linear Lie algebra $\mathrm{sl}(2,\mathbb R)$.
\end{corollary}

\begin{proof}
The correspondence $-\mathcal P^y\mapsto\mathrm P^y$, $-\mathcal D\mapsto\mathrm D$ and $-\mathcal K\mapsto\mathrm K$
by linearity extends to the Lie algebra homomorphism~$\varphi$ from~$\mathfrak f$
to the Lie algebra~\smash{$\Upsilon_{\mathfrak f}^{(-)}$}
associated with the associative algebra~$\Upsilon_{\mathfrak f}$,
$\varphi\colon\mathfrak f\to\smash{\Upsilon_{\mathfrak f}^{(-)}}$.
By the universal property of the universal enveloping algebra~$\mathfrak U(\mathfrak f)$,
the Lie algebra homomorphism~$\varphi$ extends to the (unital) associative algebra homomorphism
$\hat\varphi\colon\mathfrak U(\mathfrak f)\to\Upsilon_{\mathfrak f}$,
i.e., $\varphi=\hat\varphi\circ\iota$ as homomorphisms of vector spaces,
where $\iota\colon\mathfrak f\to\mathfrak U(\mathfrak f)$ is the canonical embedding of $\mathfrak f$ in $\mathfrak U(\mathfrak f)$.
Since the algebra $\Upsilon_{\mathfrak f}$ is generated by $\varphi(\mathfrak f)$,
the homomorphism $\hat\varphi$ is surjective.
Fixed an ordering of the basis elements of~$\mathfrak f$,
in view of the Poincar\'e--Birkhoff--Witt theorem, the homomorphism $\hat\varphi$ maps
the corresponding Poincar\'e--Birkhoff--Witt basis of $\mathfrak U(\mathfrak f)$ to a basis of $\Upsilon_{\mathfrak f}$.
Therefore, $\hat\varphi$ is an isomorphism.
The isomorphism $\mathfrak f\simeq\mathrm{sl}(2,\mathbb R)$
implies the isomorphism $\mathfrak U(\mathfrak f)\simeq\mathfrak U(\mathrm{sl}(2,\mathbb R))$.
\end{proof}

\
Given a linear system of differential equations~$\mathcal L$
and its linear generalized symmetry with the associated linear differential operator~$\mathscr Q$,
we have that \emph{$\mathscr Qh$ is a solution of~$\mathcal L$ whenever $h$ is}.
This can be interpreted as generating solutions from known one using linear generalized symmetries~\cite{kova2024b};
see~\cite{shte1987a,shte1989a} for first examples of such generation.
In view of Conjecture~\ref{conj:GenSyms} and Lemma~\ref{lem:BasisUpsilon},
the most general formula for generating solutions of the fine Kolmogorov backward equation~\eqref{eq:FineFP}
using its linear generalized symmetries from the seed solution $u=h(t,x,y)$ takes the form
\begin{gather}\label{eq:FineFPSolutionGenerationByLieSymOps}
\tilde u=\sum_{\alpha\in\mathbb N_0^{\,\,3}}\lambda_\alpha\mathbf Q^\alpha h(t,x,y).
\end{gather}
Here all but finitely many (real) constants $\lambda_\alpha$ are equal to zero,
and any ordering in $\{\mathcal P^y,\mathcal D,\mathcal K\}$ can be fixed.

However, depending on~$h$ and~$\lambda_\alpha$, $\alpha\in\mathbb N_0^{\,\,3}$,
this formula can result in a known or even the zero solution.
Let us analyze which solutions generated according~\eqref{eq:FineFPSolutionGenerationByLieSymOps}
from Lie invariant solutions of~\eqref{eq:FineFP} may be of interest.
It suffices to consider only the solutions that are invariant
with respect to one of the inequivalent one-dimensional subalgebras of~$\mathfrak g^{\rm ess}$
listed in Corollary~\ref{cor:GessIneqSubalgebras}.
Such solutions were described in Section~\ref{sec:LieReductions}.
We exclude the subalgebra~$\mathfrak s_{1.2}$ as that admitting a trivial zero space of invariant solutions.

In view of Corollary~\ref{cor:IsomToUnivEnv},
the operator~$\hat{\mathrm P}^t$ generates the center~$\mathrm Z$ of the algebra $\Upsilon_{\mathfrak f}$.
Therefore, for any subalgebra~$\mathfrak s$ of $\mathfrak g^{\rm ess}$,
any polynomial of~$\hat{\mathrm P}^t$ maps the space of $\mathfrak s$-invariant solutions of~\eqref{eq:FineFP} to itself.
It is also clear that given an element~$Q$ of~$\mathfrak g^{\rm ess}$
and the corresponding operator~$\mathrm Q$ from~$\Upsilon_{\mathfrak f}$,
the action by any element from the principal left ideal $(\mathrm Q):=\Upsilon_{\mathfrak f}\mathrm Q$
on any $\langle Q\rangle$-invariant solution of~\eqref{eq:FineFP} results in the zero solution.
More generally,  any element of $(\mathrm Q)+\mathrm C(\mathrm Q)$, i.e.,
the sum of any elements from the principal left ideal $(\mathrm Q)$
and from the centralizer~$\mathrm C(\mathrm Q)$ of~$\mathrm Q$ in~$\Upsilon_{\mathfrak f}$
(obviously containing~$\mathrm Z$)
maps the space of $\langle Q\rangle$-invariant solutions of~\eqref{eq:FineFP} to itself.

\medskip\par\noindent
1.1$^\mu$. For the above reason, the Lie-symmetry operator $\hat{\mathrm P}^t-\mu$,
which corresponds to the Lie-symmetry vector field $-\mathcal P^t-\mu\mathcal I$ of~\eqref{eq:FineFP},
commutes with any monomial~$\mathbf Q^\alpha$.
Hence the action by any element from~$\Upsilon_{\mathfrak f}$ maps
the set of $\mathfrak s_{1.1}^\mu$-invariant solutions to itself.
This implies that the solution families~\eqref{eq:SolFineFPHeat1} and~\eqref{eq:SolFineFPHeatWithZ2-2Pot}
of the equation~\eqref{eq:FineFP} cannot be extended
using the generation of its solutions by the action of its Lie-symmetry operators
according to~\eqref{eq:FineFPSolutionGenerationByLieSymOps}.
Another manifestation of the above phenomenon is that
each Lie-symmetry vector field of~\eqref{eq:SolFineFPHeat1} induces
a Lie-symmetry vector field of reduced equation~1.1$^\mu$.
Hence the algebra~$\Upsilon_{\mathfrak f}$ induces
the entire algebra generated by Lie-symmetry operators of reduced equation~1.1$^\mu$ if $\mu\ne-\frac3{16}$
and a proper subalgebra of the latter algebra otherwise.

\medskip

The rest of the cases are more complicated to analyze.
We make only particular remarks on them.

\medskip\par\noindent
1.3$^{1,\mu}$.
For the same reason as above, the operator $\mathrm P^y+\hat{\mathrm P}^t-\mu$
commutes with any polynomial in $\mathrm P^y$ and~$\hat{\mathrm P}^t$.
This is why the action by such polynomials
maps the space of $\mathfrak s_{1.3}^{1,\mu}$-invariant solutions to itself.
It can be shown that the centralizer of~$\mathrm P^y$ in~$\Upsilon_{\mathfrak f}$
and thus the centralizer of~$\mathrm P^y+\hat{\mathrm P}^t-\mu$ in~$\Upsilon_{\mathfrak f}$
coincide with the space of these polynomials.

\medskip\par\noindent
1.4$^{\delta'}$.
It is clear that for any $\delta'\in\{0,1\}$ the operator $\mathrm P^y-\delta'$ commutes with $\mathrm P^y$.
Hence the space of $\mathfrak s_{1.4}^{\delta'}$-invariant solutions is mapped to itself
at least by the action by any polynomial in $\mathrm P^y$ and~$\hat{\mathrm P}^t$.

Let $\delta'=0$.
In view of Remark~\ref{rem:BasisUpsilon},
we have the representation $\Upsilon_{\mathfrak f}=\Gamma_1\dotplus\Gamma_2$.
Since $\Gamma_1$ is the space of polynomials in \smash{$(\hat{\mathrm P}^t,\mathrm P^y,\mathrm D)$},
any element of it acts within the space of $\mathfrak s_{1.4}^0$-invariant solutions.
This is why to generate new solutions of~\eqref{eq:FineFP} from solutions of the form~\eqref{eq:SolFineFPHeat2}
within the framework discussed,
up to linearly combining solutions,
it suffices to iteratively act by the operator~$\mathrm K$,
see Section~\ref{sec:GenReductions} for explicit formulas and the relation to generalized reductions.

\medskip\par\noindent
1.5$^{\nu,\mu}$.
Recall that this case corresponds to the two subalgebra families listed in Corollary~\ref{cor:GessIneqSubalgebras},
\smash{$\{\mathfrak s_{1.5}^{\nu',\mu}\}$} and $\{\mathfrak s_{1.6}^{\nu}\}$,
see the beginning of Section~\ref{sec:LieReductions}.
Since $\mathrm K\mathrm D=(\mathrm D-1)\mathrm K$ and $\mathrm P^y\mathrm D=(\mathrm D+1)\mathrm P^y$, we have
\[
\mathrm K^{\alpha_1}(\mathrm P^y)^{\alpha_2}\mathrm D^{\alpha_3}(\mathrm D+\nu\hat{\mathrm P}^t-\mu)=
(\mathrm D+\nu\hat{\mathrm P}^t-\mu-\alpha_1+\alpha_2)\mathrm K^{\alpha_1}(\mathrm P^y)^{\alpha_2}\mathrm D^{\alpha_3},
\]
In other words, the monomial $\mathrm K^{\alpha_1}(\mathrm P^y)^{\alpha_2}\mathrm D^{\alpha_3}$
maps the space of $\{\mathfrak s_{1.5}^{\nu,\mu}\}$-invariant solutions
to the space of \smash{$\{\mathfrak s_{1.5}^{\nu,\tilde\mu}\}$}-invariant solutions, where $\tilde\mu=\mu+\alpha_1-\alpha_2$.
Hence any element of~$\Upsilon_{\mathfrak f}$ maps
the space spanned by the \smash{$\{\mathfrak s_{1.5}^{\nu,\tilde\mu}\}$}-invariant solutions
with $\mu$ running through~$\mathbb R$ into itself.

\medskip\par\noindent
1.7$^{\nu,\mu}$.
We can only state that the action by each polynomial in $\mathrm P^y+\mathrm K$ and~$\hat{\mathrm P}^t$
maps the space of $\mathfrak s_{1.7}^{\nu,\mu}$-invariant solutions to itself
since such polynomials commute with the operator $\mathrm P^y+\mathrm K+\nu\hat{\mathrm P}^t-\mu$.

\section{Generalized reductions}\label{sec:GenReductions}

Linear generalized symmetries of a linear system~$\mathcal L$ of differential equations
in the unknown functions $u=(u^1,\dots,u^m)$
can be used not only for generating solutions from known ones that are obtained by other methods
but also for constructing solutions from scratch.
The basis of the method of linear generalized reduction is the fact that
if $\mathrm Q$ is a symmetry operator of~$\mathcal L$,
then the differential constraint $\mathrm Q u=0$ is compatible with~$\mathcal L$.
Nevertheless, the procedure of finding generalized invariant solutions is not as simple
as its counterpart for Lie-invariant solutions at all.
To construct a generalized ansatz, it is necessary to construct
the general solution of the system $\mathrm Q u=0$ in an explicit form,
which is possible only in particular cases.
Moreover, it often happens that the solutions invariant with respect to a generalized symmetry
can be obtained in a much simpler way by generation from known solutions.

We present two peculiar examples of generalized reductions of the equation~\eqref{eq:FineFP},
and the first series of generalized invariant solutions
is related to the results of Section~\ref{sec:SolutionGeneration}.

The equation~\eqref{eq:FineFP} admits generalized symmetries $\big({\rm D}_y^nu\big)\p_u$, $n\in\mathbb N$,
which are generated from its elementary Lie symmetry $u\p_u$
by acting with the recursion operators  $\mathrm P^y={\rm D}_y$
associated with its Lie symmetry $-\mathcal P^y=-\p_y$.
This choice of a Lie symmetry for generating generalized symmetries is justified
by the consideration of codimension-one Lie reduction~$1.4^0$ in Section~\ref{sec:LieReductions}.
The differential constraint ${\rm D}_y^nu=0$ can be easily integrated,
which gives the corresponding generalized ansatz
\begin{gather}\label{eq:GenAnsatz1}
u=\sum_{s=0}^{n-1}w^s(t,x)\frac{y^s}{s!}.
\end{gather}
Here and in what follows $w^s$ are sufficiently smooth functions of their arguments,
the index~$s$ runs from~0 to~$n-1$, and $\epsilon=\sgn x$.

Substituting the generalized ansatz~\eqref{eq:GenAnsatz1} into~\eqref{eq:FineFP}
and splitting the resulting equation with respect to~$y$,
we obtain the system $w^s_t=x^2w^s_{xx}-xw^{s+1}$ with $w^n:=0$.
The point transformation
$z_1=t$, $z_2=\ln|x|$, $\tilde w^s={\rm e}^{\frac14t}|x|^{-1/2}w^s$
maps it to the system
\begin{gather*}
\mathcal H_n\colon\quad \tilde w^s_1=\tilde w^s_{22}-\epsilon{\rm e}^{z_2}\tilde w^{s+1}
\end{gather*}
with $\tilde w^n:=0$,
which is easier for integrating since its last equation is the (1+1)-dimensional linear heat equation.
This is why the general solution of the system~$\mathcal H_n$ can be expressed
in terms of $n$~arbitrary solutions of this equation.
We can obtain a particular solution of the system~$\mathcal H_n$
by prolonging any solution of~$\mathcal H_{n'}$ with $n'<n$ by zeros to~$\tilde w^s$, $s=n',\dots,n-1$,
and hence such solutions can be considered as inessential.

\begin{lemma}
The solutions of~$\mathcal H_n$ that are essential up to subtracting the prolonged solutions of~$\mathcal H_{n'}$, $n'=1,\dots,n-1$,
and thus have $\tilde w^{n-1}\ne0$, take the form
\begin{gather*}
\tilde w^s=\epsilon^s{\rm e}^{(n-s-1)z_2}
\left(\prod_{k=n-s}^{n-1}(2k\p_2+k^2)\right)\theta(z_1,z_2),\quad
\end{gather*}
where $\theta=\theta(z_1,z_2)$ is an arbitrary solution of the heat equation.
\end{lemma}

\begin{proof}
In fact, we should construct a particular solution of the system of $n-1$ first equations of~$\mathcal H_n$
for an arbitrary solution of the last equation of~$\mathcal H_n$.
We use the induction with respect to the parameter $n$.
The induction base $n=0$ is clear.

Let now $n>0$.
We reduce the construction of the solution to that for the system~$\mathcal H_{n-1}$.
We substitute an ansatz $\tilde w^0={\rm e}^{(n-1)z_2}\theta(z_1,z_2)$
with a solution $\theta$ of the heat equation $\theta_1=\theta_{22}$
to the first equation of the system $\mathcal H_n$ and thus derive the representation
$\tilde w^1={\rm e}^{(n-2)z_2}\tilde\theta$,
where $\tilde\theta=\epsilon(2(n-1)\p_2+(n-1)^2)\theta$.
Since $\epsilon(2(n-1)\p_2-(n-1)^2\theta\p_\theta)$ is a Lie-symmetry vector field
of the heat equation $\theta_1=\theta_{22}$,
the function $\tilde\theta$ is a solution of this equation as well.
The subsystem of the $n-1$ last equations of~$\mathcal H_n$
coincides with the system $\mathcal H_{n-1}$ up to the index shift by 1, $s'=s-1$.
By the induction hypothesis, this subsystem has the particular solution
\begin{gather*}
\tilde w^s=\epsilon^{s-1}{\rm e}^{(n-1-s)z_2}
\left(\prod_{k=n-s}^{n-2}(2k\p_2+k^2)\right)\tilde\theta(z_1,z_2),\quad
\end{gather*}
where $s$ ranges through the set $\{1,\dots,n-1\}$.
Substituting the expression of~$\tilde\theta$ in terms of~$\theta$ into this solution
and completing the range of~$s$ by zero in view of the ansatz for~$\tilde w^0$, we obtain
the required representation for particular solutions of the system $\mathcal H_n$,
where the component $\tilde w^{n-1}=\epsilon^{n-1}\big(\prod_{k=1}^{n-1}(2k\p_2+k^2)\,\big)\theta$
runs through the solution set of the linear heat equation if the function~$\theta$ does,
see \cite[Lemma~7]{kova2023b}.
\end{proof}

Therefore, the essential solutions of the equation~\eqref{eq:FineFP}
that are associated with the generalized ansatz~\eqref{eq:GenAnsatz1} take the form
\begin{gather}\label{eq:GenSolution1}
\solution
u={\rm e}^{-\frac14 t}|x|^{n-\frac12}\sum_{s=0}^{n-1}\frac{y^sx^{-s}}{s!}
\left(\prod_{k=n-s}^{n-1}(2kx\p_x+k^2)\right)\theta(t,\ln|x|),
\end{gather}
where $\theta=\theta(z_1,z_2)$ is an arbitrary solution
of the heat equation $\theta_1=\theta_{22}$.

The solutions of the form~\eqref{eq:GenSolution1} admit an interpretation in terms of the generation of new solutions
by the action of Lie-symmetry operators on Lie-invariant solutions.

\begin{lemma}\label{lem:FineFPLieSymOpsCommutators}
For any $n\in\mathbb N$, \
$\displaystyle
{\rm D}_y^{n+1}{\rm K}^n
=\left(\prod_{k=1}^n({\rm D}_y{\rm K}+2k{\rm D}_y+k^2+k)\right){\rm D}_y$.
\end{lemma}

\noprint{
\begin{proof}
We prove the formula by induction on $n$.
The base case $n=1$ directly follows from the commutation relations between involved differential operators,
$
{\rm D}_y^2{\rm K}
={\rm D}_y({\rm K}{\rm D}_y+2{\rm D})
=\left({\rm D}_y{\rm K}+2({\rm D}+1)\right){\rm D}_y
$.
Supposing the induction hypothesis for~$n$, we obtain the formula for~$n+1$.
For any $k\in\mathbb N$, we have
\begin{gather*}
{\rm D_y}\left({\rm D}_y{\rm K}+2\sum_{s=1}^{k}({\rm D}+s)\right)
={\rm D}_y^2{\rm K}+2\sum_{s=1}^{k}({\rm D}_y{\rm D}+s{\rm D}_y)
\\ \qquad{}
=\left({\rm D}_yK+2({\rm D}+1)\right){\rm D}_y+2\sum_{s=1}^{k}({\rm D}+s+1){\rm D}_y
=\left({\rm D}_y{\rm K}+2\sum_{s=1}^{k+1}({\rm D}+s)\right){\rm D_y}
\end{gather*}
We use the derived formula for proving the induction step,
\begin{align*}
{\rm D}_y^{n+2}{\rm K}^{n+1}
={\rm D}_y({\rm D}_y^{n+1}{\rm K}^n){\rm K}
&={\rm D}_y\left(\prod_{k=1}^n\left({\rm D}_y{\rm K}+2\sum_{s=1}^{k}({\rm D}+s)\right)\right){\rm D}_y{\rm K}
\\
&=\left(\prod_{k=1}^n\left({\rm D}_y{\rm K}+2\sum_{s=1}^{k+1}({\rm D}+s)\right)\right){\rm D}_y^2{\rm K}
\\
&=\left(\prod_{k=1}^n\left({\rm D}_y{\rm K}+2\sum_{s=1}^{k+1}({\rm D}+s)\right)\right)\left({\rm D}_y{\rm K}+2({\rm D}+1)\right){\rm D}_y
\\
&=\left(\prod_{k=1}^{n+1}\left({\rm D}_y{\rm K}+2\sum_{s=1}^k({\rm D}+s)\right)\right){\rm D}_y,
\end{align*}
which completes the induction.
\end{proof}
}

Lemma~\ref{lem:FineFPLieSymOpsCommutators} implies that
${\rm K}^nu$ satisfies the equation ${\rm D}_y^{n+1}u=0$
if $u$ is a solution of the equation ${\rm D}_yu=0$.
There is even a more precise assertion.

\begin{theorem}
For any $n\in\mathbb N$, the family of solutions of the form~\eqref{eq:GenSolution1}
that are invariant with respect to the generalized symmetries $({\rm D}_y^nu)\p_u$
can be expressed in terms of the iterative action by the Lie-symmetry operator~${\rm K}$
on $\mathfrak s_{1.4}^0$-invariant solutions of~\eqref{eq:FineFP},
\smash{$u={\rm K}^{n-1}\big({\rm e}^{-\frac14t}|x|^{1/2}\theta(t,\ln|x|)\big)$}.
\end{theorem}

This theorem is proved by induction with respect to the parameter~$n$, which involves cumbersome computations.

The second example of generalized reductions of the equation~\eqref{eq:FineFP}
is related to its recursion operator ${\rm D}_t-\mu$,
which is associated with the Lie-symmetry vector field $-\p_t-\mu u\p_u$.
Thus, this equation admits the generalized symmetries $(({\rm D}_t-\mu)^nu)\p_u$, $n\in\mathbb N$.
We can expect a nontrivial output of applying the generalized reductions with respect to these symmetries
due to their association with Lie reduction $1.1^\mu$.
The corresponding generalized ansatz is
\begin{gather*}
u={\rm e}^{\mu t}\sum_{s=0}^{n-1}w^s(x,y)\frac{t^s}{s!}.
\end{gather*}
Substituting it into~\eqref{eq:FineFP}
and splitting the resulting equation with respect to~$t$,
we obtain the system $xw^s_y=x^2w^s_{xx}-\mu w^s-w^{s+1}$ with $w^n:=0$.
It is mapped
by the point transformation $z_1=\epsilon y$, \smash{$z_2=2|x|^{\frac12}$}, \smash{$\tilde w^s=|x|^{-\frac14}w^s$}
to the system
\begin{gather*}
\mathcal S^\mu_n\colon\quad
\tilde w^s_1=\tilde w^s_{22}-\frac{4\mu+3/4}{z_2^2}\tilde w^s-\frac4{z_2^2}\tilde w^{s+1}
\end{gather*}
with $\tilde w^n:=0$,
whose last equation is the linear $(1+1)$-dimensional heat equation with the inverse square potential
$(4\mu+3/4)z_2^{-2}$.

The Lie-symmetry vector field $-\mathcal P^t-\mu\mathcal I$
belongs to the center~$\mathfrak z$ of the algebra~$\mathfrak g^{\rm ess}$
and thus the associated Lie-symmetry operator $\hat{\mathrm P}^t-\mu$
commutes with all elements of~$\Upsilon_{\mathfrak f}$.
This is why the solutions of~\eqref{eq:FineFP} that are invariant with respect to $(({\rm D}_t-\mu)^nu)\p_u$
cannot be related to generating solutions with iterative action of Lie-symmetry operators,
cf.\ paragraph~1.1$^\mu$ in Section~\ref{sec:SolutionGeneration}.

Consider the particular case $\mu=0$ and $n=2$.
The system~$\mathcal S^0_2$ takes the form
\begin{gather*}
\tilde w^0_1-\tilde w^0_{22}+\frac3{4z_2^2}\tilde w^0=-\frac4{z_2^2}\tilde w^1,\quad
\tilde w^1_1-\tilde w^1_{22}+\frac3{4z_2^2}\tilde w^1=0.
\end{gather*}
We consider the first equation of~$\mathcal S^0_2$ as an inhomogeneous linear equation
with respect to~$w^0$ and look for its particular solution in the form
$\tilde w^0=-2z_2^{-1}v$, where the new unknown function
$v=v(z_1,z_2)$ satisfies the equation $v_1-v_{22}-\frac14z_2^{-2}v=0$.
Substituting this ansatz to the first equation of~$\mathcal S^0_2$,
we derive the representation $\tilde w^1=v_2-\frac12z_2^{-1}v$,
which means that $\tilde w^1$ is the Darboux transformation
of~$v$ with respect to the solution $v=|z_2|^{1/2}$,
see \cite{matv1991A,popo2010a} for the required definitions and results.
This Darboux transformation exactly relates the equation for~$v$ and the second equation of~$\mathcal S^0_2$.
As a result, up to linearly combining with $\mathfrak s_{1.1}^0$-invariant solutions,
we construct the following family of solutions of the equation~\eqref{eq:FineFP}:
\begin{gather}\label{eq:GenSolution2}
\solution
u=|x|^{1/4}\Big(
t\vartheta^{3/4}_2(\varepsilon y,2\sqrt{|x|})-
(\tfrac14t-1)|x|^{-1/2}\vartheta^{3/4}(\varepsilon y,2\sqrt{|x|})
\Big),
\end{gather}
where $\vartheta^{3/4}=\vartheta^{3/4}(\tilde z_1,\tilde z_2)$ is an arbitrary solution of the equation
$\vartheta_1=\vartheta_{22}-\frac34\tilde z_2^{-2}\vartheta$,
and $\vartheta^{3/4}_2$ is its derivative with respect to its second argument.
Recall that wide families of solutions for the latter equation were obtained in~\cite[Section~A]{kova2023a}.

\section{Conclusion}\label{sec:Conclusion}

\looseness=1
The symmetry properties of the fine Kolmogorov backward equation~\eqref{eq:FineFP} are prominent
within both the general class~$\bar{\mathcal F}$
of $(1+2)$-dimensional ultraparabolic linear partial differential equations
and the subclass~$\mathcal F$ of such equations with the standard ultraparabolic part and power diffusivity,
see Section~\ref{sec:Introduction}.
As a Kolmogorov backward equation, it also appears in various applications.
These facts motivated us to thoroughly carry out the extended symmetry analysis of the equation~\eqref{eq:FineFP}.

The initial point of this study is the computation of the maximal Lie invariance algebra~$\mathfrak g$
of the equation~\eqref{eq:FineFP} in Section~\ref{sec:FineFPMIA}.
The algebra~$\mathfrak g$ splits over its infinite-dimensional ideal $\mathfrak g^{\rm lin}$,
which is associated with the linear superposition of solutions of~\eqref{eq:FineFP},
$\mathfrak g=\mathfrak g^{\rm ess}\lsemioplus\mathfrak g^{\rm lin}$.
The subalgebra $\mathfrak g^{\rm ess}$ is five-dimensional and reductive.
It is isomorphic to the algebra
${\rm sl}(2,\mathbb R)\oplus2A_1$ and thus to ${\rm gl}(2,\mathbb R)\oplus A_1$.
Up to the equivalence with respect to point transformation,
the above symmetry properties completely single out the equation~\eqref{eq:FineFP}
among the members of the class~$\bar{\mathcal F}$.

In Section~\ref{sec:FineFPPointSymGroup} we have computed the point symmetry pseudogroup~$G$ of the equation~\eqref{eq:FineFP}
using the direct method reinforced by the our knowledge of transformational properties of the entire class~$\bar{\mathcal F}$.
According to~\cite[Theorem~1]{kova2023a}, the class~$\bar{\mathcal F}$ is normalized in the usual sense.
Under the identification of the pseudogroup~$G$ with the vertex group of~\eqref{eq:FineFP}
in the equivalence groupoid~$\mathcal G^\sim_{\bar{\mathcal F}}$ of~$\bar{\mathcal F}$,
this implies that the pseudogroup~$G$ is contained in the natural projection
of the equivalence group~$G^\sim_{\bar{\mathcal F}}$ of~$\bar{\mathcal F}$
to the space with the coordinates~$(t,x,y,u)$.
The application of the reinforced direct method has led to a significant simplification
of the system of determining equations for the transformation components,
which is still a coupled overdetermined nonlinear system of partial differential equations.
We have successfully found the general solution to this complicated system
and constructed a nice representation of the elements of the group~$G$ in Theorem~\ref{thm:FineFPSymGroup}.

Furthermore, after redefining the multiplication in the pseudogroup~$G$
following the approaches from \cite{kova2023a,kova2023b},
we have exhaustively studied its structure.
The pseudogroup~$G$ splits over its normal pseudosubgroup $G^{\rm lin}$ that is associated
with the linear superposition of solutions of~\eqref{eq:FineFP}, $G=G^{\rm ess}\ltimes G^{\rm lin}$,
and due to using the modified multiplication, we can choose the first component~$G^{\rm ess}$ of the splitting
to be a group, which is, moreover, a five-dimensional Lie group.
This analysis required some notions from the general theory of pseudogroups of transformations,
which we also have revisited and enhanced in Section~\ref{sec:Pseudogroups}.
The group $G^{\rm ess}$ admits the factorization $G^{\rm ess}=F\times Z$,
where $F\simeq{\rm PSL}^\pm(2,\mathbb R)$ and $Z\simeq(\mathbb R^2,+)\times\mathbb Z_2$,
which corresponds to the Lie algebra decomposition $\mathfrak g^{\rm ess}=\mathfrak f\oplus\mathfrak z$.
Another convenient factorization for the study of the structure of~$G^{\rm ess}$ is $G^{\rm ess}=H\times P$,
where $H\simeq{\rm GL}(2,\mathbb R)$ and $P\simeq(\mathbb R,+)\times\mathbb Z_2$.
In the process of constructing the decompositions of~$G^{\rm ess}$,
we used the results given in Section~\ref{sec:StructureOfSLpm}
about the structure of the two-by-two matrix groups~${\rm SL}^\pm_2(\mathbb R)$,
${\rm PSL}^\pm_2(\mathbb R)$ and~${\rm GL}_2(\mathbb R)$.

The accurate description of the structure of $G^{\rm ess}$
has allowed us to correctly classify the subalgebras of the algebra~$\mathfrak g^{\rm ess}$ in Section~\ref{sec:Subalgebras}
modulo two equivalences.
One equivalence is induced by the action of the inner automorphism group ${\rm Inn}(\mathfrak g^{\rm ess})$,
which is the same as the action of $G^{\rm ess}_{\rm id}$ on $\mathfrak g^{\rm ess}$ by the adjoint representation.
The other equivalence is defined by the action of the Lie group~$G^{\rm ess}$ on its Lie algebra~$\mathfrak g^{\rm ess}$.

The classification of $G^{\rm ess}$-inequivalent subalgebras of~$\mathfrak g^{\rm ess}$
allowed us to exhaustively describe Lie reductions of the equation~\eqref{eq:FineFP} in Section~\ref{sec:LieReductions}.
Due to the comprehensive study of hidden symmetries of~\eqref{eq:FineFP},
we have shown that it suffices to carry out only inequivalent Lie reductions of codimension one.
For each of these reductions,
the corresponding reduced equation is a linear homogeneous second-order evolution equation in two independent variables,
which has been mapped to a linear (1+1)-dimensional heat equation with a potential depending only
on the space-like independent variable.
The peculiar reductions occurred when the corresponding potential is trivial or the inverse square one
lead to constructing wide families of explicit solutions~\eqref{eq:SolFineFPHeat1}--\eqref{eq:SolFineFPHeatWithZ2-2Pot}
of~\eqref{eq:FineFP}.

Section~\ref{sec:SolutionGeneration} is devoted to generating solutions of the equation~\eqref{eq:FineFP}
using its Lie and generalized symmetries.
We have discussed various situations where the (repeated) action by Lie-symmetry operators
may or may not give new solutions, which is the first discussion of this kind in the literature.
An important role in the discussion was played by the nonobvious fact
that the operator associated with the right-hand side of~\eqref{eq:FineFP}
can be interpreted as a Casimir operator of the Levi factor
$\mathfrak f=\langle\mathcal P^y,\mathcal D,\mathcal K\rangle$ of~$\mathfrak g^{\rm ess}$.

Carrying out some particular generalized reductions of the equation~\eqref{eq:FineFP} in Section~\ref{sec:GenReductions},
we have constructed wide families of its solutions
parameterized by an arbitrary finite number of arbitrary solutions of the (1+1)-dimensional linear heat equation~\eqref{eq:GenSolution1}
or one or two arbitrary solutions of (1+1)-dimensional linear heat equations with inverse square potentials~\eqref{eq:GenSolution2}.
To the best of our knowledge, these are the first exact solutions of the equation~\eqref{eq:FineFP} in the literature,
and they can be important for testing numerical solvers for this equation.
We have treated the above results in the context of those on generating solutions.

We reasonably conjectured the structure of the algebra
of generalized symmetries of the equation~\eqref{eq:FineFP} in~\cite{kova2024b}.
Nevertheless, we expect that the proof of this conjecture is sophisticated and cumbersome
and it should be the subject of a separate study.
The classification of reduction modules of the equation~\eqref{eq:FineFP}
is also of interest, especially singular ones in view of its ultraparabolicity.
Besides, one could study potential symmetries of this equation although
up to now there are no meaningful examples of potential symmetries
of multidimensional partial differential equations (both linear and nonlinear) in the literature.

\appendix

\section{Pseudogroups}\label{sec:Pseudogroups}

The study of symmetries of differential equations is local by its nature.
A very nice structure that captures this feature and arises in many other applications
is given by local Lie transformation groups \cite[Definition~1.20]{olve1993A}.
Each of such objects comes along with a smooth action on a manifold.
However, this notion does not suffice to cover all the cases,
including those when the corresponding set of transformations cannot be parameterized by a finite number of parameters,
which is common, e.g., for the point symmetry groups of all linear and many nonlinear systems of partial differential equations.
Moreover, a local Lie transformation group is not a group (e.g., as the collection of M\"obius transformations~\cite{laws1998a})
but rather a groupoid (in the algebraic sense).
This is why it is more convenient to use the more general notion of pseudogroup.

We first present an (enhanced) definition of pseudogroup of transformations of a topological space
which is often used in the literature \cite[Section~1.1]{rein1983A}
and can be further specified for topological spaces with additional structures (i.e., {\it local\'e}) such as manifolds,
and then modify it.
A {\it pseudogroup of transformations} on a topological space~$X$ is a set~$G$
of homeomorphisms between nonempty open subsets of~$X$ with the operation of map composition
that satisfy the following properties:
\begin{enumerate}\itemsep=1ex
\item
The identity map of~$X$ belongs to $G$.
\item
If $\Phi\in G$, then its restriction to any open nonempty subset of its domain belongs to~$G$.
\item
If $\Phi\in G$, then $\Phi^{-1}\in G$.
\item
If $\Phi_i\in G$, $\Phi_i\colon U_i\to V_i$, $i=1,2$, with $V_1\cap U_2\ne\varnothing$,
then $\Phi_2\circ\Phi_1\colon\Phi_1^{-1}(V_1\cap U_2)\to \Phi_2(V_1\cap U_2)$ belongs to $G$.
\item
Given any collection $\{U_\alpha\}$ of open sets of~$X$,
a homeomorphism from~$\cup_\alpha U_\alpha$ onto an open set of~$X$ belongs to~$G$
if and only if for every~$\alpha$, its restriction to~$U_\alpha$ belongs to~$G$.
\end{enumerate}

In view of property~4, the operation in~$G$ is partially defined,
i.e., not all pairs of elements of~$G$ can be composed.
This is why the (algebraic) concept of groupoid comes in handy.

\begin{definition}
An {\it (algebraic) groupoid} is a nonempty set~$G$ together with
a unary operation (inverse) $^{-1}\colon G\to G$
and a partial binary operation%
\footnote{%
``$*$'' is not a usual (total) binary operation since it is not necessarily defined for all pairs of elements of~$G$.}
(multiplication) $*\colon G\times G\rightharpoonup G$ satisfying the following properties:
\begin{description}\itemsep=0ex
\item[{\it Associativity}$\colon$]
If $a*b$ and $b*c$ are defined, then $(a*b)*c$ and $a*(b*c)$ are defined and equal to each other.
Moreover, $a*b$ and $(a*b)*c$ are defined if and only if $b*c$ and $a*(b*c)$ are defined as well.
\item[{\it Inverse}$\colon$]
For every $a\in G$, $a*a^{-1}$ and $a^{-1}*a$ are defined.
\item[{\it Identity}$\colon$]
If $a*b$ is defined, then $a*b*b^{-1}=a$ and $a^{-1}*a*b=b$.
\end{description}
\end{definition}

As defined above, any pseudogroup of transformations on a topological space is a groupoid,
where the map composition is considered as the multiplication.
The partially defined operation makes the study of the corresponding structure involved and challenging.
However, by modifying the above definition of a pseudogroup,
one can endow a pseudogroup with the structure of an {\it inverse monoid}~\cite{laws1998A},
which makes analyzing it more pleasurable.
The modification consists in treating homeomorphisms between open subsets of a topological space~$X$
as partially defined functions from~$X$ to~$X$,
allowing the empty function to be an element of a pseudogroup
and setting the composition of partial functions as the pseudogroup operation.

\begin{definition}\label{def:Pseudogroup}
Given a topological space~$X$,
a {\it pseudogroup~$G$ of transformations} on~$X$ is a collection
of partial homeomorphisms on~$X$, whose domains%
\footnote{%
The (natural) domain of a partial function $\Phi\colon X\rightharpoonup Y$ is the domain of $\Phi$ viewed as a function.
A \emph{partial homeomorphism} between topological spaces is a homeomorphism between their subspaces viewed as a partial function.
}
and images are open subsets of~$X$ and that satisfy the following properties:
\begin{enumerate}\itemsep=0ex
\item
The identity function of $X$ is in $G$.
\item
The restriction of any element $\Phi\in G$ to any open set contained in its domain is also an element of $G$.
\item
The composition $\Phi_1\circ\Phi_2$ of any pair of elements of $G$ is also in $G$.
\item
The inverse of an element~$\Phi$ is in $G$.
\item
Given any collection $\{U_\alpha\}$ of open sets of~$X$,
a partial homeomorphism of~$X$ with the domain~$\cup_\alpha U_\alpha$ onto an open set of~$X$ belongs to~$G$
if and only if for every~$\alpha$, its restriction to~$U_\alpha$ belongs to~$G$.
\end{enumerate}
\end{definition}

It is clear that Definition~\ref{def:Pseudogroup} can be generalized to any local\'e.

\begin{definition}(\cite[p.~6]{laws1998A})
An {\it inverse semigroup} is a nonempty set $S$ with an associative binary operation $*\colon S\times S\to S$
such that for every $a\in S$, there exists a unique $a^{-1}\in S$,
called the {\it inverse} of $a$, satisfying $a^{-1}*a*a^{-1}=a^{-1}$ and $a*a^{-1}*a=a$.
An {\it inverse monoid} $M$ is an inverse semigroup with \emph{neutral element} $e\in M$, $e*a=a*e=a$ for all $a\in M$.
\end{definition}

See more~\cite[Chapter~1]{laws1998A} for an excellent presentation of the basic theory of inverse semigroups and inverse monoids.

Properties~1, 3 and~4 of Definition~\ref{def:Pseudogroup} ensure
that any pseudogroup is an inverse monoid.
Throughout the paper, when referring to pseudogroups,
we mean them in the sense of Definition~\ref{def:Pseudogroup}.

An element~$a$ of a set with a binary operation~$*$ is called an idempotent if $a*a=a$.
One can completely characterize the set of idempotents in any pseudogroup.

\begin{lemma}
An element of a pseudogroup is an idempotent if and only if it is a restriction of the identity map.
\end{lemma}

\begin{proof}
Given a pseudogroup~$G$ of transformations on~$X$, every ${\rm id}_U$, where $U$ is an open subset of~$X$, is an idempotent.
Conversely, consider an idempotent $\Phi\in G$, $\Phi^2=\Phi$.
Therefore, ${\Phi^2=\Phi\circ{\rm id}_{\mathop{\rm dom}\Phi}}$, and thus $\Phi={\rm id}_{\mathop{\rm dom}\Phi}$.
\end{proof}

The required notions are introduced by analogy with their counterparts for groups
with modification caused by the fact that any pseudogroup is rather an inverse monoid.

A subset of a pseudogroup~$G$ that itself is a pseudogroup of transformations between open sets of a topological space~$X$
is called a \emph{pseudosubgroup} of~$G$.
Given any pseudogroup~$G$, the set of its idempotents,
$G^{\rm id}:=\{{\rm id}_U^{}\mid U\text{ is an open subset of}~X\}$, and the entire~$G$ itself
are pseudosubgroups of~$G$ called the \emph{improper pseudosubgroups} of~$G$.
They are the minimal and the maximal pseudosubgroups of~$G$ with respect to the set inclusion,
and the other pseudosubgroups of~$G$ are called \emph{proper}.
If $G=G^{\rm id}$, then there are no proper pseudosubgroups in~$G$.

Given a subgroup~$H$ of~$G$,
we call a pseudosubgroup~$N$ of~$G$ \emph{invariant with respect to~$H$} if
$\Phi N\Phi^{-1}\subseteq N$ for any $\Phi\in H$,
i.e., the pseudosubgroup~$N$ is stable under the conjugation by any element~$\Phi$ of~$H$.
It is obvious that then $\Phi N\Phi^{-1}=N$ or, equivalently, $\Phi N=N\Phi$ for any $\Phi\in H$.

If $H$ is a subgroup of~$G$,
$N$ is a pseudosubgroup of~$G$ that is invariant with respect to~$H$
$G=NH$ and $N\cap H=\{{\rm id}_X\}$,
then we say that $G$ is the \emph{semidirect product} of $H$ acting on~$N$,
writing $G=H\ltimes N$ or $G=N\rtimes H$, or that $G$ \emph{splits over}~$N$,
or even that $G$ is the semidirect product of~$N$ and~$H$.

In a similar way, one can try to introduce the more general notions of normal pseudosubgroup
and splitting a pseudogroup over such a pseudosubgroup.
We can suggest the following definitions
although their naturality requires a thorough study.
A \emph{normal pseudosubgroup}~$N$ of~$G$ is such a pseudosubgroup of~$G$ that
$\Phi N=N\Phi$ for any $\Phi\in G$, i.e,
for any $\Phi\in G$ and for any $\Psi\in N$ we have
$\Phi\circ\Psi=\hat\Psi\circ\Phi$ and $\Psi\circ\Phi=\Phi\circ\check\Psi$
for some $\hat\Psi,\check\Psi\in N$.
If $H$ and $N$ are respectively a pseudosubgroup and a normal pseudosubgroup of~$G$,
$G=NH$ and $N\cap H=G^{\rm id}$,
then $G=HN$ as well and we say that $G$ is the \emph{semidirect product} of $H$ acting on $N$,
writing $G=N\rtimes H$ or $G=H\ltimes N$, or that $G$ \emph{splits over}~$N$,
or even that $G$ is the semidirect product of $N$ and $H$.
We can weaken these notions even more.
A pseudosubgroup~$N$ of~$G$ is called a \emph{right normal pseudosubgroup} of~$G$
if $\Phi N\subseteq N\Phi$ for any $\Phi\in G$, i.e,
for any $\Phi\in G$ and for any $\Psi\in N$ we have
$\Phi\circ\Psi=\hat\Psi\circ\Phi$ for some $\hat\Psi\in N$.
Given a pseudosubgroup~$H$ and a right normal pseudosubgroup of~$G$ with $G=NH$ and $N\cap H=G^{\rm id}$,
we say that $G$ is the \emph{right semidirect product} of $H$ acting on $N$,
writing $G=N\rtimes H$, or that $G$ \emph{right splits over}~$N$,
or even that $G$ is the right semidirect product of $N$ and $H$.
Left normal pseudosubgroups and left semidirect product are defined analogously.

\section{Two-by-two matrix groups as semidirect products}\label{sec:StructureOfSLpm}

In this section we present a proof that the group~${\rm SL}^\pm(2,\mathbb R)$ of two-by-two matrices with determinant $\pm1$
is a semidirect product of its subgroup ${\rm SL}(2,\mathbb R)$ and the subgroup isomorphic to $\mathbb Z_2$,
and as a corollary we obtain the analogous statement for the groups ${\rm PSL}^\pm(2,\mathbb R)$ and ${\rm GL}(2,\mathbb R)$.
Despite the fact that the statement looks classical, we did not find the proof in the literature.

\begin{proposition}
${\rm SL}^\pm(2,\mathbb R)\simeq{\rm SL}(2,\mathbb R)\rtimes_\varphi\mathbb Z_2$,
where $\varphi$ is the homomorphism from $\mathbb Z_2$ to the group of automorphisms of ${\rm SL}(2,\mathbb R)$, $\varphi\colon\mathbb Z_2\to{\rm Aut}({\rm SL}(2,\mathbb R))$,
whose value~$\varphi(\bar 1)$ on the generator~$\bar 1$ of $\mathbb Z_2$ is the involution
\begin{gather*}
\begin{pmatrix}
a&b\\
c&d
\end{pmatrix}
\mapsto
\begin{pmatrix}
a&-b\\
-c& d
\end{pmatrix}.
\end{gather*}
\end{proposition}

\begin{proof}
Consider the group homomorphism $\det\colon{\rm SL}^\pm(2,\mathbb R)\to{\rm U}(\mathbb Z)$,
where ${\rm U}(\mathbb Z)$ denotes the group of units of the ring $\mathbb Z$, ${\rm U}(\mathbb Z)=\{1,-1\}$,
i.e., the set of all invertible elements of the ring $\mathbb Z$ with the operation given by the multiplication in $\mathbb Z$.
The group ${\rm U}(\mathbb Z)$ is isomorphic to $\mathbb Z_2$.
The following sequence is exact
\[
0 \rightarrow{\rm SL}(2,\mathbb R)\stackrel\iota\hookrightarrow
{\rm SL}^\pm(2,\mathbb R)\stackrel\det\twoheadrightarrow {\rm U}(\mathbb Z)\rightarrow0,
\]
\noprint{
\[
\begin{tikzcd}
0\arrow[r]
&{\rm SL}(2,\mathbb R)\arrow[r,hook, "\iota"]
&{\rm SL}^\pm(2,\mathbb R)\arrow[r,twoheadrightarrow,"\det"]
&{\rm U}(\mathbb Z)\arrow[r]
&0,
\end{tikzcd}
\]
}
where $\iota$ is the inclusion homomorphism.
	
We need to show that the above sequence splits.
In other words, there exists a right section of $\det$,
$\nu\colon{\rm U}(\mathbb Z)\to{\rm SL}^\pm(2,\mathbb R)$,
such that ${\rm SL}^\pm(2,\mathbb R)={\rm SL}(2,\mathbb R)\nu({\rm U}(\mathbb Z))$ and ${\det}\circ\nu={\rm id}_{\rm U(\mathbb Z)}$.
To satisfy the first condition, we should require that $\det\nu(-1)=-1$,
and since $\nu(-1)^2={\rm diag}(1,1)$, the matrix $\nu(-1)$ is diagonalizable with the eigenvalues $\pm1$,
so the Jordan normal form for $\nu(-1)$ is ${\rm diag}(1,-1)$.
	
Taking the map $\nu\colon{\rm U}(\mathbb Z)\to\rm{SL}^\pm(2,\mathbb R)$
defined on the generator of ${\rm U}(\mathbb Z)$ by $-1\mapsto{\rm diag}(1,-1)$ as the canonical homomorphism,
we get that ${\rm SL}^\pm(2,\mathbb R)$ is the semidirect product of its subgroups ${\rm SL}(2,\mathbb R)$ and $\nu({\rm U}(\mathbb Z))$,
${\rm SL}^\pm(2,\mathbb R)={\rm SL}(2,\mathbb R)\rtimes\nu({\rm U}(\mathbb Z))$.
The conjugation in ${\rm SL}(2,\mathbb R)$ by the matrix $\nu(-1)$ gives the desired involution $\varphi(\bar 1)$.
\end{proof}

\begin{corollary}\label{cor:PSL+-}
${\rm PSL}^\pm(2,\mathbb R)\simeq{\rm PSL}(2,\mathbb R)\rtimes_\varphi\mathbb Z_2$,
where $\varphi$ is a homomorphism from $\mathbb Z_2$ to the automorphism group of ${\rm PSL}(2,\mathbb R)$,
$\varphi\colon\mathbb Z_2\to{\rm Aut}({\rm PSL}(2,\mathbb R))$,
that is given (on the generator of $\mathbb Z_2$) by the involution
\begin{gather*}
\varphi(\bar 1)=\begin{pmatrix}
a&b\\
c&d
\end{pmatrix}{\rm Z(SL(2,\mathbb R))}
\mapsto
\begin{pmatrix}
a&-b\\
-c& d
\end{pmatrix}{\rm Z(SL(2,\mathbb R))}.
\end{gather*}
\end{corollary}

\begin{corollary}
${\rm GL}(2,\mathbb R)\simeq{\rm GL}^+(2,\mathbb R)\rtimes_\varphi\mathbb Z_2$,
where $\varphi$ is a homomorphism from $\mathbb Z_2$ to the automorphism group of ${\rm GL}^+(2,\mathbb R)$,
$\varphi\colon\mathbb Z_2\to{\rm Aut}({\rm GL}^+(2,\mathbb R))$,
that is given (on the generator of $\mathbb Z_2$) by the involution
\begin{gather*}
\varphi(\bar 1)=\begin{pmatrix}
a&b\\
c&d
\end{pmatrix}
\mapsto
\begin{pmatrix}
a&-b\\
-c& d
\end{pmatrix}.
\end{gather*}
\end{corollary}

\section*{Acknowledgements}
The authors are grateful to Alexander Bihlo and Vyacheslav Boyko for valuable discussions.
This research was undertaken thanks to funding from the Canada Research Chairs program and the NSERC Discovery Grant program.
It was also supported in part by the Ministry of Education, Youth and Sports of the Czech Republic (M\v SMT \v CR)
under RVO funding for I\v C47813059.
ROP expresses his gratitude for the hospitality shown by the University of Vienna during his long stay at the university.
The authors express their deepest thanks to the Armed Forces of Ukraine and the civil Ukrainian people
for their bravery and courage in defense of peace and freedom in Europe and in the entire world from russism.

\end{document}